\documentclass[12pt, draftclsnofoot, onecolumn]{IEEEtran}
\usepackage{setspace}
\usepackage{amsmath}
\usepackage{amssymb}
\usepackage{amsfonts}
\usepackage{epsfig}
\usepackage{cite}
\usepackage{enumerate}
\usepackage{float}
\usepackage{afterpage}
\usepackage{makeidx}
\usepackage{graphicx}
\usepackage{epstopdf}
\usepackage{caption}
\usepackage{subcaption}
\usepackage{multicol}
\usepackage{multirow}
\usepackage{amsthm}
\newtheorem{theorem}{Theorem}

\newcommand{\nn}{\nonumber}
\usepackage{color}
\usepackage{bigints}
\usepackage{array}
\begin{document}	
\title{Performance of Energy-Buffer Aided  Incremental Relaying in Cooperative Networks}
\author{{Dileep Bapatla, \emph{Student Member, IEEE} and Shankar Prakriya, \emph{Senior Member, IEEE}}\thanks{Dileep Bapatla and Shankar Prakriya are with the Department of Electrical Engineering, Indian Institute of Technology Delhi, (IIT Delhi) New	Delhi 110016, India (e-mail: Dileep.Bapatla@ee.iitd.ac.in, shankar@ee.iitd.ac.in).}	}	
\maketitle	
\begin{abstract}	
In this paper, we consider a two-hop cooperative network with a direct link based on an energy harvesting (EH) decode-and-forward relay. The energy-buffer equipped relay harvests energy from the ambience, and uses the harvest-store-use (HSU) architecture.  
Since it is known that using a discrete-state Markov chain to model the energy buffer is inaccurate even for moderate number of states, we use a discrete-time continuous-state space Markov chain instead. We derive the limiting distribution of energy for both incremental on-off policy (IOFP) and the incremental best-effort policy (IBEF), and use them to obtain expressions for outage probability and throughput. The corresponding expressions for non-incremental signalling follow as a special case. We show that stable buffers using IBEP harness a diversity of two as compared to those using IOFP, which attain a diversity of one. However, while buffers using IBEF are consequently more reliable than those with IOFP, their throughput performance is only marginally superior. Simulation results are presented to validate the derived analytical expressions.
\end{abstract}
\IEEEpeerreviewmaketitle
\section{Introduction}
 Relays have been incorporated into several modern communication standards (e.g. LTE-A) due to their promise in increasing range and reliability of wireless networks \cite{M.N.Tehrani2014}. Relays are needed when the direct link is shadowed, when the source has power limitations,  or when diversity gain needs to be harnessed.  For this reason, cooperative relay networks have been extensively studied in recent years. Regenerative or decode-and-forward (DF) relays, as well as non-regenerative amplify-and-forward (AF) relays have been widely studied.
 \par Conventionally, relays are assumed to be equipped with a power source or a battery  \cite{S.Bi2016}. Of late however, there has been considerable research interest in relays that are powered by energy harvesting (EH). In applications like mobile sensor networks and wireless body area networks \cite{H.Mosavat-Jahromi2017},  use of a nearby node to serve as a relay to communicate to a distant node is well motivated. In a 5G pico-cell communication framework, use of EH by nodes serving as relays is being explored. In these applications, the node serving as a relay is often battery powered, and use of its battery energy to provide relaying services is undesirable. This is because battery life-times are typically short, and replacement of batteries is often difficult (and in some applications, impossible). Use of EH  to make the devices battery-less is therefore desirable. In other scenarios, the node serving as a relay may seek to supplement the harvested energy with as little battery energy as possible to ensure long battery lifetimes \cite{Modem}. Use of relays powered by EH is also motivated by green energy considerations - telecom is known to be a large consumer of energy. Analysis of performance of cooperative communication links with EH relays  is clearly of considerable  interest.
 \par  EH is readily possible from natural sources like light, wind, vibration, and ambient radio frequency (RF) signals  \cite{S.Sudevalayam2011}.  Simultaneous wireless information and power transfer (SWIPT), which exploits the possibility of RF signals carrying both energy and information simultaneously,  has attracted considerable research interest \cite{X.Zhou2013}. However, while they enable communication with a battery-less relay, they are not energy efficient (green) due to path-loss in practical channels. The same is true of  power beacons \cite{Y.Ma}.  In this paper, we focus on relays powered by ambient sources \cite{S.Yin}.
\par  Despite the apparent similarity between energy and data buffers \cite{N.Zlatanov2013}, \cite{B.Kumar}, there is one fundamental difference - while data buffers are essentially discrete, energy is a continuous variable. However, virtually all literature in this area is based on discrete-state energy buffers \cite{Moradian}.  It has been shown that for accurate representation, the number of states needs to be  quite large \cite{Y.Gu2015} (close to $200$ \cite{vucetic2018},\cite{I.Krikidis2012}), which make the approach difficult and cumbersome. Also,  it is difficult to draw physical insights into performance of such systems with this approach. 
\par When the direct channel from source to destination is not shadowed, it is advantageous to combine the signals from source and relay optimally. In addition,  use of incremental relaying to improve spectral efficiency is well motivated.   
While \cite{N.T.Van} is based on harvest-use (HU) architecture,  \cite{Z.Li} uses a discrete state space Markov chain to model the buffer (the relay harvests energy from source).  \par Recently, it has been shown that a discrete-time continuous-state Markov chain model can be used to model the energy buffer accurately in point-to-point links \cite{R.Morsi2017}.  Performance depends on the policy used to operate the buffer. With the   best-effort energy storage  policy  \cite{R.Morsi2015}, the node transmits by drawing constant $M$ amount of energy when sufficient energy exists, and drawing entire energy in the buffer otherwise.  With the on-off storage policy \cite{R.Morsi2014}, the node transmits using  $M$ amount of energy when it is available, and remains silent otherwise.  Performance of single-hop wireless powered communication  with both these models was analyzed, and derivation for the asymptotic distribution of energy in the buffer was presented in \cite{R.Morsi2017}. In this paper, we analyze the performance of cooperative links with buffered relays using IBEP and IOFP, and derive expressions for asymptotic distributions of energy. Expressions for non-incremental signalling follow as a special case. To the best of our knowledge there has been no work on performance of cooperative links with energy buffers modeled using a discrete-time continuous-state Markov chain\footnote{A conference version of this paper dealing  only with IBEP was submitted to IEEE VTC 2019 \cite{Dileep}. }. The significant contributions are as follows:
\begin{enumerate}
	\item Unlike all prior work on EH relays with energy buffers  that are based on discrete state Markov chain  \cite{vucetic2018},\cite{I.Krikidis2012}, in this paper we analyze performance of cooperative links using a continuous state space Markov chain to model the energy buffer. Also, unlike most other works on EH \cite{A.A.Nasir}, we do not ignore the direct path from source to destination, and optimally combine the direct and relayed signals at the destination.
	\item For incremental signalling (IBEP as well as IOFP), we derive an expression for the limiting (steady-state) distribution of energy in the energy buffer, and establish conditions for its existence. 
	\item Expressions are derived for outage probabilities and throughput with both IBEP and IOFP. Expressions for non-incremental signalling follow as a special case. We compare performance with the HU architecture and  direct  (relay-less) transmission. 
\item We establish that the IBEP attains a diversity order of two, while the IOFP attains a diversity of one, making links with the former policy more reliable. However, the throughput performance with IBEP is only marginally superior to that with IOFP.
\end{enumerate}
\subsection{Notations:} 
${\cal CN}\sim \left(\mu,\sigma^{2}\right)$	denotes the  complex Gaussian distribution with mean $\mu$ and variance $\sigma^{2}$, $\sim$ denotes ``distributed as", and $|x|$ denotes the absolute value of $x$. $ \Pr \left\lbrace A\right\rbrace $ indicates probability of the event $A$. $\mathcal{W}\left(\cdot\right) $ is the principal branch of the Lambert-W function.  $ \min \left( a,b\right) $ denotes minimum of $a,b$. $\mathbb{E} \left\lbrace \cdot\right\rbrace $ denotes the expectation operator. $F_{X}\left( x\right) $ and $f_{X}\left( x\right) $ denote the cumulative distribution function (CDF) and the probability density function (PDF) of random variable $X$. 
\section{System Model}
The network (as depicted in Fig.~\ref{coop}) consists of source S, a DF relay R, and a destination D. While the source S  is equipped with a power supply, the relay R is powered solely by harvested ambient energy \cite{S.Yin}. Quasi-static Rayleigh fading is assumed. Let   $d_{ab}$ denote the normalized distance between nodes with $a\in [S,R]$, and $b\in [R,D]$. Denote by $h_{_{SR}}(i)\sim{\cal CN}(0,d_{_{SR}}^{-\alpha})$, 	$h_{_{RD}}(i)\sim{\cal CN}(0,d_{_{RD}}^{-\alpha})$,  and $h_{_{SD}}(i)\sim{\cal CN}(0,d_{_{SD}}^{-\alpha})$ the channels in the $i^{th}$ signalling interval between S and R, R and D, and S and D respectively (where $\alpha$ is the path-loss exponent). 

\begin{figure}[h!]
	\centering
	\includegraphics[scale=1.4]{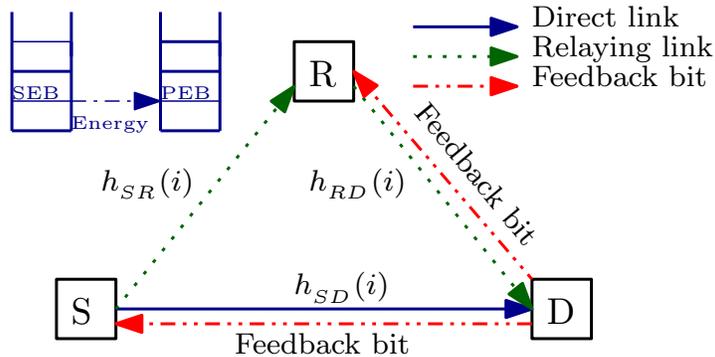}
	\caption{Energy buffer-aided incremental relaying cooperative communication. If D always sends a negative acknowledgment (through the feedback bit), the non-incremental case results.}
	\label{coop}
	\vspace*{-0.8 cm}
\end{figure}



\par   Two types of EH architectures have been proposed in literature  \cite{S.Sudevalayam2011}.
 With the HSU (Harvest-store-use) architecture,  the relay is equipped with a rechargeable device (battery for example) that can store and release energy. Since storage devices cannot charge and discharge at the same time \cite{S.Luo2013},\cite{B.Zhang2016}, harvested energy is stored in a secondary energy buffer (SEB), which is typically a super-capacitor, and transferred to the primary energy buffer (PEB), which is typically a battery, at the end of the signalling interval in a negligible amount of time \cite{B.Zhang2016}. On the other hand, with HU (Harvest-use) architecture,  all the energy that is harvested is used in the same signalling interval (for example architectures based on super-capacitors).
It is the HSU architecture that is of primary interest in this paper. For comparison purpose, we also consider direct (relay-less) transmission from source to destination, as well as the HU architecture. The reason for use of the feedback bit will become clear in what follows. 
 \par For both HSU and HU architectures, the signalling interval is divided into two phases of equal duration. The overall signalling interval $T$  is normalized to unity for convenience.  In the first phase (phase I) of normalized duration $1/2$, the source S transmits unit-energy information symbols $x_S(i)$ at rate $R_{0}$ with power $P_{_{S}}$ to the relay R and destination D. The received signals $y_{_{SD}}(i)$ and $y_{_{SR}}(i)$ at D and R in phase I of the $i\textsuperscript{th}$ signalling interval are given by: 
\begin{eqnarray}
	&y_{_{SD}}(i)=\sqrt{P_{_{S}}}h_{_{SD}}(i) x_{_{S}}(i) +n_{_{SD}}(i), &\text{Phase I}\\
	&y_{_{SR}}(i)=\sqrt{P_{_{S}}}h_{_{SR}}(i) x_{_{S}}(i) +n_{_{SR}}(i), & \text{Phase I}
	\label{1-3}
\end{eqnarray}
where $n_{_{SD}}(i),n_{_{SR}}(i) \sim {\cal CN}\left(0,\sigma^{2}\right)$ are additive noise samples. The instantaneous link signal to noise ratios (SNRs) in the $i\textsuperscript{th}$ signalling interval at D and R in the first phase are given by:
\begin{eqnarray}
&\gamma_{_{SD}}(i)=\displaystyle\frac{P_{_{S}}|h_{_{SD}}(i)|^{2}}{\sigma^{2}}, ~~\gamma_{_{SR}}(i)=\frac{P_{_{S}}|h_{_{SR}}(i)|^{2}}{\sigma^{2}}. 
\label{link_snr}
\end{eqnarray}	
Let:
\begin{eqnarray}
\Gamma_{th}^{\prime} &=& \left\{ \begin{array}{ll} \Gamma_{th}=2^{2R_{0}}-1 & \text{Incremental Signalling,}\\
\infty & \text{Non-incremental Signalling.}
 \end{array}\right.
\end{eqnarray}
It will become apparent later that this definition allows us to analyze the performance of incremental and non-incremental relaying in a unified fashion.\\
{\bf Case $\gamma_{_{SD}}(i)\geq \Gamma^{\prime}_{th}$}:\  If  the received SNR at D is greater than (or equal to) the threshold $\Gamma^{\prime}_{th}$ (signalling is successful in the first phase itself),  D sends a positive acknowledgment (ACK), which is received by both R and S. On receipt of ACK, S directly transmits {\em a new set of symbols}  to D in the second phase (phase II), making the overall SNR $\gamma_D(i)=\gamma_{_{SD}}(i)$. Note that an ACK can never be sent for non-incremental signalling.\\
{\bf Case $\gamma_{_{SD}}(i)<\Gamma^{\prime}_{th}$}:\ 
 If received SNR at D is less than the threshold in the first phase, D sends a negative acknowledgment (NACK), which is received by both S and R.   Clearly, a NACK is always sent for non-incremental signalling. In this case, the fixed DF  \cite{J.N.Laneman} R decodes $x_{_{S}}(i)$, re-encodes, and then transmits the information as unit-energy symbols $x_R(i)$ at rate $R_0$ with power $P_{_{R}}(i)$ to D using the energy harvested (hence the dependence of $P_{_{R}}(i)$ on $i$). Clearly, the signal $y_{_{RD}}(i)$ received by D in this phase is given by:
\begin{eqnarray}
&y_{_{RD}}(i)=\sqrt{P_{_{R}}(i)}h_{_{RD}}(i) x_{_{R}}(i) +n_{_{RD}}(i), & \text{Phase II}
\end{eqnarray}
where $n_{_{RD}}(i) \sim {\cal CN}\left(0,\sigma^{2}\right).$ Clearly, $\gamma_{_{RD}}(i)= P_{_{R}}(i)|h_{_{RD}}(i)|^{2}/\sigma^{2}.$  The signals received from S and R are combined using maximal ratio combining (MRC) at D. The overall SNR $\gamma_D(i)$ at D is given by $\gamma_{_{SD}}(i)+\gamma_{_{RD}}(i)$.  
\section{Limiting distribution of Energy with HSU Architecture}
In this section, we derive expressions for the limiting distributions of energy in the PEB with both IBEP and IOFP, assuming it to be of infinite-size. Since the energy harvested is typically small, this assumption  is not limiting. Those with non-incremental signalling follow as special cases (with $\Gamma_{th}^{\prime}=\infty$), and will not be written explicitly.

\subsection{Limiting Distribution of IBEP}
\par The incoming harvested energy  $X(i)$ is assumed to be an exponential random  variable (as in \cite{S.Luo2013},\cite{B.Zhang2016}) with PDF $f_{X}(x)$.  Denote by $B(i)$ the energy level in the buffer in the $i^{th}$ signalling interval. Let $M$ denote the energy drawn for every relay transmission. The PEB's buffer  update equation is then given by:
\begin{eqnarray}
B(i+1) & = &  B(i)+X(i) \hspace{1.5 cm} \left( \gamma_{_{SD}}(i) \geq \Gamma^{\prime}_{th}\right) , \nonumber\\ 
B(i+1) & = & B(i)-M+X(i) \hspace{0.6 cm} (\gamma_{_{SD}}(i)<\Gamma^{\prime}_{th})\cap (B(i) \geq M), \nonumber\\ 
B(i+1) & = & X(i) \hspace{3 cm} (\gamma_{_{SD}}(i)<\Gamma^{\prime}_{th})\cap (B(i)<M).
\label{BI2}
\end{eqnarray}
Let $ W_{1},W_{2},W_{3} $ and $ W_{4} $  be defined as 
\hspace{0.1 cm}  $W_{1}=\dfrac {\sigma^2d_{_{RD}}^\alpha}{2M}$,\hspace{0.01 cm}  $W_{2}=\dfrac {\sigma^2d_{_{SD}}^\alpha}{P_{_{S}}}$,\hspace{0.01 cm}   $W_{3}=\dfrac {\sigma^2d_{_{RD}}^\alpha}{2}$ \hspace{0.01 cm} and \hspace{0.01 cm} $W_{4}=\dfrac {\sigma^2d_{_{SR}}^\alpha}{P_{_{S}}}$.
 Let
\begin{align}
 \phi_{inc} = \frac{M\left( 1-e^{-W_{2}\Gamma^{\prime}_{th}} \right)}{\mathbb{E}{\{X(i)}\}}=M\lambda_{1}\left( 1-e^{-W_{2}\Gamma^{\prime}_{th}} \right),
\label{phinc}
\end{align} 
where $\mathbb{E}{\{X(i)}\}=\displaystyle \frac{1}{\lambda_{1}}$. Note that the equivalent constant $\phi_{ninc}$ in the non-incremental case is given by $\phi_{ninc}=M\lambda_1$ (using $\Gamma_{th}^{\prime}=\infty$). The limiting distributions when $ \phi_{inc} >1 $ is presented in Theorem~\ref{Theorem_Limiting_PDF}, and the case when $ \phi_{inc} \leq 1 $ is discussed in Theorem~\ref{Theorem_Limiting_PDF not exists}.

\begin{theorem}\label{Theorem_Limiting_PDF} 
For an infinite-size PEB based on the IBEP as in \eqref{BI2}, the limiting energy distribution  exists when $\phi_{inc}=M\lambda_1 (1-\exp(-W_2\Gamma^{\prime}_{th})) >1$, and its  limiting CDF $G(x)$ satisfies the following integral equation:
 \begin{align}
G(x)=\int_{u=0}^{x}\left[ e^{-W_{2}\Gamma^{\prime}_{th}}f_{X}\left( x-u\right)G\left(u\right)+\left( 1-e^{-W_{2}\Gamma^{\prime}_{th}}\right)f_{X}\left( u\right)G\left( x+M-u\right) \right] du,
\label{limcdfintegr}
\end{align}
whose solution is given by:
\begin{align}
G(x)=\left( 1-\exp\left( -Z_{1}x\right) \right),
\label{limcdf}
\end{align}
where
$$ \qquad {Z_{1}}  =  \displaystyle \frac{\mathcal{W}\left( -\left( 1-e^{-W_{2}\Gamma^{\prime}_{th}} \right) \lambda_{1}M\exp\left( -\left( 1-e^{-W_{2}\Gamma^{\prime}_{th}}\right) \lambda_{1}M \right) \right) }{M}  
+\left( 1-e^{-W_{2}\Gamma^{\prime}_{th}} \right) \lambda_{1}.$$ 
\end{theorem}
\begin{proof}
Refer to Appendix \ref{Appendix_Limiting_PDF}. 
\end{proof} 
\begin{theorem}\label{Theorem_Limiting_PDF not exists}
 For an infinite-size PEB based on the IBEP as in \eqref{BI2}, the limiting energy distribution does not exist when $\phi_{inc}=M\lambda_1 (1-\exp(-W_2\Gamma^{\prime}_{th})) \leq 1$, and $M$ amount of energy is almost always available for relay transmission in the PEB.  Relay transmit power $P_{R}(i)$ equals $2M$  ($P_{_{R}}\left( i\right) =\dfrac{M}{\left( 1/2\right)}=2M$) almost surely.
\end{theorem}
\begin{proof}
	Refer to Appendix \ref{Appendix_Limiting_PDF not exists}. 
\end{proof}
\subsection{Limiting Distribution of IOFP}
 Unlike the on-off policy described for point-to-point links described in \cite{R.Morsi2014}, the relay R here becomes silent: 1) when it receives an ACK from the destination, or 2) when it does not have $M$ amount of energy in its PEB but NACK is received from D. With IOFP, the energy buffer update equation of the PEB is given by:
\begin{eqnarray}
B(i+1) & = &  B(i)+X(i) \hspace{1.9 cm} (\gamma_{_{SD}}(i) \geq \Gamma^{\prime}_{th}), \nonumber\\ 
B(i+1) & = & B(i)-M+X(i) \hspace{1 cm} (\gamma_{_{SD}}(i)<\Gamma^{\prime}_{th})\cap (B(i)\geq M), \nonumber\\ 
B(i+1) & = &B(i)+X(i) \hspace{2.1 cm} (\gamma_{_{SD}}(i)<\Gamma^{\prime}_{th})\cap (B(i)<M).
\label{BIonoff}
\end{eqnarray}
The limiting distribution when $ \phi_{inc} > 1 $ is presented in Theorem \ref{Theorem_Limiting_PDF On-Off}, and the case when $ \phi_{inc} \leq 1 $ is first discussed in Theorem~\ref{on-off not exists}. 
Let $\bar{I}$ be an indicator variable defined as: $\bar{I}=1 \hspace*{0.2 cm} \text{if} \hspace*{0.2 cm} B\left( i\right) \geq M $, and
$\bar{I}=0$ otherwise.

\begin{theorem}\label{on-off not exists}
	For the infinite-size PEB based on the IOFP as in \eqref{BIonoff}, the limiting distribution does not exist when $\phi_{inc} \leq 1$.  $M$ amount of energy is almost always available for transmission in the PEB, and  relay transmit power  $P_R\left( i\right) =\dfrac{M}{\left( 1/2\right)}=2M$ almost surely. 
\end{theorem}
\begin{proof}
	Proof is along lines similar to that of Theorem \ref{Theorem_Limiting_PDF not exists} with R having transmit energy $E_{_{R}}(i)=M\bar{I}$, and is therefore omitted.
\end{proof}

\begin{theorem}\label{Theorem_Limiting_PDF On-Off}
 For an infinite-size PEB based on the IOFP as in \eqref{BIonoff}, the limiting distribution of infinite-size PEB exists when $\phi_{inc} >1$, and its limiting CDF $G(x)$ must satisfy following integral equation:
 \begin{align}
G(x)=
\begin{cases}
e^{-W_{2}\Gamma^{\prime}_{th}}\displaystyle\int_{u=0}^{x}F_{X}\left( x-u\right)g\left( u\right)du 
+\left( 1-e^{-W_{2}\Gamma^{\prime}_{th}}\right) \left[\int_{u=M}^{M+x}F_{X}\left( x-u+M\right)g\left( u\right)du \right . \\
\left.+\displaystyle \int_{u=0}^{x}F_{X}\left( x-u\right)g\left(u\right)du \right] \hspace{2.9 cm} 0 \leq x < M \\
e^{-W_{2}\Gamma^{\prime}_{th}}\displaystyle\int_{u=0}^{x}F_{X}\left( x-u\right)g\left( u\right)du 
+\left( 1-e^{-W_{2}\Gamma^{\prime}_{th}}\right)\left[\int_{u=M}^{M+x}F_{X}\left( x-u+M\right)g\left( u\right)du \right . \\
\left.+\displaystyle \int_{u=0}^{M}F_{X}\left( x-u\right)g\left(u\right)du \right] \hspace{2.9 cm}     x \geq M
\end{cases}
\label{limcdfintsonoff}
\end{align}
Its corresponding limiting PDF $g(x)$ is given by:
\begin{align}
g(x)=\begin{cases}
\dfrac{\left(1-e^{Qx}\right)}{M}  \hspace{4.2 cm} 0 \leq x < M \\
\dfrac{-Qe^{Qx}}{M\left(Q+\lambda_{1}\left( 1-e^{-W_{2}\Gamma^{\prime}_{th}}\right)\right) }  \hspace{2 cm} x \geq M,
\end{cases}
\label{limpdfsplifin}
\end{align}
where $Q$ is given by:
$$Q =  -\frac{\mathcal{W}\left( -\left( 1-e^{-W_{2}\Gamma^{\prime}_{th}} \right) \lambda_{1}M\exp\left( -\left( 1-e^{-W_{2}\Gamma^{\prime}_{th}}\right) \lambda_{1}M \right) \right) }{M}  
-\left( 1-e^{-W_{2}\Gamma^{\prime}_{th}} \right) \lambda_{1}. $$
\end{theorem}
For  existence of the limiting distribution, $Q$ should be negative (i.e. $Q < 0$).
\begin{proof}
	Refer to Appendix \ref{Appendix_Limiting_PDF On-Off}. 
\end{proof}

\section{Outage and Throughput Analysis}
In this section, expressions are derived for outage and throughput  with the HSU and HU architectures.	
\subsection{HSU Architecture with IBEP}
 With incremental relaying, the link is  not in outage when $\gamma_{_{SD}}(i)\geq \Gamma^{\prime}_{th}$. When $\gamma_{_{SD}}(i)< \Gamma^{\prime}_{th}$, the link is once again not in outage when both $\gamma_{_{SR}}(i)$ and $\gamma_{_{SD}}(i)+\gamma_{_{RD}}(i)$ are greater than or equal to $\Gamma_{th}$.  The outage probability $P_{out}^{HSU}$ for incremental relaying is given by \cite{S.S.Ikki}:
\begin{align}
P_{out}^{HSU}=\Pr\left\lbrace \gamma_{_{SD}}(i)<\Gamma^{\prime}_{th}\}\Pr\{\min\left( \gamma_{_{SR}}(i),\gamma_{_{RD}}(i)+\gamma_{_{SD}}(i)\right) < \Gamma_{th}\mid(\gamma_{_{SD}}(i)<\Gamma^{\prime}_{th})\right\rbrace,
\label{outin}
\end{align}
where $\gamma_{_{SR}}(i)$, $\gamma_{_{SD}}(i)$ are given by \eqref{link_snr} and $\gamma_{_{RD}}(i)=P_{_{R}}(i)|h_{_{RD}}(i)|^{2}/\sigma^{2}$ .\\
The throughput expression $\tau^{HSU}$ for incremental relaying with fixed DF protocol is given by:
\begin{align} 
\tau^{HSU}= & \hspace{0.1 cm} R_{0}\Pr\left\lbrace \gamma_{_{SD}}\left( i\right) \geq \Gamma^{\prime}_{th}\right\rbrace \nonumber
+\frac {R_{0}}{2}\Pr\left\lbrace \gamma_{_{SD}}\left( i\right) <\Gamma^{\prime}_{th}\right\rbrace \nn \\
&\times \Pr\left\lbrace \min\left( \gamma_{_{SR}}\left( i\right) ,\gamma_{_{RD}}\left( i\right) +\gamma_{_{SD}}\left( i\right)\right) \geq \Gamma_{th}\mid \left( \gamma_{_{SD}}(i)<\Gamma^{\prime}_{th}\right) \right\rbrace .
\label{ta}
\end{align}
In \eqref{ta}, the last term can be rewritten as:
\begin{align}
&\Pr\left\lbrace \gamma_{_{SD}}\left( i\right) <\Gamma^{\prime}_{th}\right\rbrace 
\Pr\left\lbrace \min\left( \gamma_{_{SR}}\left( i\right) ,\gamma_{_{RD}}\left( i\right) +\gamma_{_{SD}}\left( i\right) \right) \geq \Gamma_{th} \mid\left( \gamma_{_{SD}}(i)<\Gamma^{\prime}_{th}\right) \right\rbrace \nn \\
&=\Pr\left\lbrace \min\left( \gamma_{_{SR}}\left( i\right) ,\gamma_{_{RD}}\left( i\right) +\gamma_{_{SD}}\left( i\right)\right) \geq \Gamma_{th} ,  \gamma_{_{SD}}(i) < \Gamma^{\prime}_{th} \right\rbrace \nn \\
&\overset{\ell}{=}\Pr\left\lbrace \gamma_{_{SD}}\left( i\right) <\Gamma^{\prime}_{th}\right\rbrace-\Pr\left\lbrace \min\left( \gamma_{_{SR}}\left( i\right) ,\gamma_{_{RD}}\left( i\right) +\gamma_{_{SD}}\left( i\right)\right) < \Gamma_{th} ,  \gamma_{_{SD}}(i)<\Gamma^{\prime}_{th} \right\rbrace \nn \\
&=\Pr\left\lbrace \gamma_{_{SD}}\left( i\right) <\Gamma^{\prime}_{th}\right\rbrace-\Pr\left\lbrace \gamma_{_{SD}}(i)<\Gamma^{\prime}_{th}\right\rbrace  \times \nn \\
& \hspace{2 cm }\Pr\left\lbrace \min(\gamma_{_{SR}}(i),\gamma_{_{RD}}(i)+\gamma_{_{SD}}(i))< \Gamma_{th}\mid(\gamma_{_{SD}}(i)<\Gamma^{\prime}_{th})\right\rbrace \nn \\
&\overset{n}{=}\Pr\left\lbrace \gamma_{_{SD}}\left( i\right) <\Gamma^{\prime}_{th}\right\rbrace-P_{out}^{HSU}=\Pr\left\lbrace \gamma_{_{SD}}(i)<\Gamma^{\prime}_{th}\right\rbrace \bigg[1-\frac {P_{out}^{HSU}}{\Pr\{\gamma_{_{SD}}(i)<\Gamma^{\prime}_{th}\}}\bigg].
\label{outs} 
\end{align}
In the above, equality $\ell$ follows from the fact that for two events $A$ and $B$, $\Pr\{A,B\}+\Pr\{\overline{A},B\}=\Pr\{B\}$, where $\overline{A}$ denotes the complement of $A$.  Equation $n$ follows from (\ref{outin}). Using \eqref{outs} in \eqref{ta} gives:
\begin{align}
\tau^{HSU} =R_{0}\Pr\{\gamma_{_{SD}}(i)\geq\Gamma^{\prime}_{th}\}+\frac {R_{0}}{2}\Pr\left\lbrace \gamma_{_{SD}}(i)<\Gamma^{\prime}_{th}\right\rbrace \bigg[1-\frac {P_{out}^{HSU}}{\Pr\{\gamma_{_{SD}}(i)<\Gamma^{\prime}_{th}\}}\bigg] ,\label{throughput}
\end{align}
where  $\Pr\left\lbrace \gamma_{_{SD}}(i)\geq\Gamma^{\prime}_{th}\right\rbrace =e^{-W_{2}\Gamma^{\prime}_{th}}. $ In the following theorem, we present expressions for outage $P_{out}^{HSU}$ using the fact that $P_R(i)$ is given by:
\begin{align}
{P_{_{R}}(i)}=
\begin{cases}
\dfrac {\min\left( B(i),M\right) }{(1/2)}=2\min\left( B(i),M\right) \hspace{1cm}  \hspace{1cm} \phi_{inc} > 1.\\
\dfrac{M}{(1/2)}=2M \hspace{6 cm} \phi_{inc} \leq 1
\end{cases}
\label{reltrb}
\end{align}
\begin{theorem}
For an infinite-size PEB based on the IBEP, the outage probability  $P_{out}^{HSU}$ with HSU architecture when $\phi_{inc} > 1$ and $\phi_{inc} \leq 1$  are given by:

\begin{align} 
P_{out}^{HSU} =&\hspace{0.1 cm} (1-e^{-W_{4}\Gamma_{th}})\left( 1-e^{-W_{2}\Gamma^{\prime}_{th}}\right)
+e^{-W_{4}\Gamma_{th}} \nn \\
 &\times\Bigg[e^{-Z_{1}M}\left\lbrace \left( 1-e^{-W_{2}\Gamma_{th}}\right)
+\frac{W_{2}\left[ e^{-W_{1}\Gamma_{th}}-e^{-W_{2}\Gamma_{th}}\right] }{(W_{1}-W_{2})}\right\rbrace \nn \\ &+(1-e^{-Z_{1}M})\left( 1-e^{-W_{2}\Gamma_{th}}\right) 
+I_{HSU}\Bigg] \hspace{0.5cm}  \hspace{0.5cm} \phi_{inc} > 1
\label{hsuoutag1}
\end{align} 
where 
\begin{align}
I_{HSU}=W_{2}Z_{1}\int_{x=0}^{M}\frac{xe^{-Z_{1}x}[e^{-W_{2}\Gamma_{th}}
-e^{-\frac{W_{3}\Gamma_{th}}{x}}] }{(xW_{2}-W_{3})}dx.
\label{integ}
\end{align}
\begin{align} 
 P_{out}^{HSU} =& \hspace{0.1 cm} (1-e^{-W_{4}\Gamma_{th}})\left( 1-e^{-W_{2}\Gamma^{\prime}_{th}}\right) \nn \\ 
&+e^{-W_{4}\Gamma_{th}} 
\left\lbrace \left( 1-e^{-W_{2}\Gamma_{th}}\right)
+\frac{W_{2}\left[ e^{-W_{1}\Gamma_{th}}-e^{-W_{2}\Gamma_{th}}\right] }{(W_{1}-W_{2})}\right\rbrace \hspace{0.8 cm} \phi_{inc} \leq 1
\label{hsuoutage2}
\end{align} 
\end{theorem}
\begin{proof}
	Refer to Appendix \ref{outage HSU arch}.
	\end{proof}
The integral $I_{HSU}$ has no closed-form expression, but is bounded,  and can be evaluated numerically. 
 The throughput $\tau^{HSU}$ is given by (\ref{throughput}), 
where $P_{out}^{HSU}$ is given by \eqref{hsuoutag1} and \eqref{hsuoutage2} for $\phi_{inc} > 1$ and $\phi_{inc} \leq 1$. As noted already, expressions for the non-incremental BEP (NIBEP) follow with $\Gamma_{th}^{\prime}=\infty$.
\subsection{HSU Architecture with IOFP}
With IOFP, the relay transmit power $P_R(i)$ is given by:
 \begin{align}
P_{_{R}}\left( i\right)=
\begin{cases}
\dfrac{M\bar{I}}{\left( 1/2\right) }=2M\bar{I}    \hspace{1.2 cm }  \phi_{inc} > 1\\
\dfrac{M}{\left( 1/2\right) }=2M   \hspace{1.5 cm } \phi_{inc} \leq 1.
\end{cases} \hspace{1 cm} 
\label{onofftra}
\end{align}
\begin{theorem} For an infinite-size buffer based on the IOFP, the outage probability $P_{out}^{HSU}$ when $\phi_{inc} > 1$  is given by:
\begin{align} 
P_{out}^{HSU} =& \hspace{0.1 cm} \left( 1-e^{-W_{4}\Gamma_{th}}\right)  \left( 1-e^{-W_{2}\Gamma^{\prime}_{th}}\right) 
+e^{-W_{4}\Gamma_{th}} \nn \\
&\times\Bigg[\frac{1}{\phi_{inc}}\left\lbrace \left( 1-e^{-W_{2}\Gamma_{th}}\right)
+\frac{W_{2}\left[ e^{-W_{1}\Gamma_{th}}-e^{-W_{2}\Gamma_{th}}\right] }{(W_{1}-W_{2})}\right\rbrace \nn \\ &+\left( 1-\frac{1}{\phi_{inc}}\right) \left( 1-e^{-W_{2}\Gamma_{th}}\right)  
\Bigg]   \hspace{1 cm}  \hspace{3.2 cm} \phi_{inc} > 1.
\label{hsuonoffoutag1}
\end{align} 
Using Theorem \ref{on-off not exists}, when $\phi_{inc} \leq 1$,  the expression for outage probability  is as in \eqref{hsuoutage2}.
\end{theorem}
\begin{proof}
	Refer to Appendix \ref{outage on-off}. 
\end{proof}

The throughput $\tau^{HSU}$ is given by (\ref{throughput}), 
where $P_{out}^{HSU}$ is given in \eqref{hsuonoffoutag1} and \eqref{hsuoutage2} for $\phi_{inc} > 1$ and $\phi_{inc} \leq 1$. As noted already, expressions for the non-incremental OFP (NIOFP) follow with $\Gamma_{th}^{\prime}=\infty$.
\subsection{HU Architecture with Incremental Relaying}
In HU Incremental (HU Inc), there is no energy buffer at the relay node. In the first phase, the  relay stores harvested energy in a super-capacitor.  In the second phase, the relay transmits to the destination with all the harvested energy if NACK is received from destination. If ACK is received from destination, the relay remains silent (the harvested energy cannot be used in future time slots due to inability of the super-capacitor to store energy over long intervals).
 With the HU architecture  the relay transmit power is given by:$
P_{_{R}}(i)=\frac {X(i)}{\left(1/2\right) }(1/2)=X(i)$.
Substituting  $P_{_{R}}(i)$ into  \eqref{outin},   the outage probability with HU architecture can be written as:
\begin{align}
P_{out}^{HU}
=& \hspace{0.1 cm}(1-e^{-W_{4}\Gamma_{th}})\left( 1-e^{-W_{2}\Gamma^{\prime}_{th}}\right)
+e^{-W_{4}\Gamma_{th}}\left[  \lambda_{1}W_{2}I_{HU}+\left( 1-e^{-W_{2}\Gamma_{th}}\right)\right] ,
\label{huouta}
\end{align}
where 
$$I_{HU}=\int_{x=0}^{\infty}\frac{xe^{-\lambda_{1}x}[e^{-W_{2}\Gamma_{th}}-e^{-\frac{2W_{3}\Gamma_{th}}{x}}]}{(xW_{2}-2W_{3})}dx. $$
	Proof is omitted due to paucity of space. The integral $I_{HU}$ has no closed form expression, but is  bounded,  and can be evaluated numerically.\\    
The throughput  $\tau^{HU}$ with HU architecture is given by  (\ref{throughput}) with $P_{out}^{HU}$ replacing $P_{out}^{HSU}$:
\begin{align}
\tau^{HU} =& \hspace{0.1 cm} R_{0}\Pr\left\lbrace \gamma_{_{SD}}(i)\geq\Gamma^{\prime}_{th}\right\rbrace +\frac {R_{0}}{2}\Pr\left\lbrace \gamma_{_{SD}}(i)<\Gamma^{\prime}_{th}\right\rbrace \left[ 1-\frac {P_{out}^{HU}}{\Pr\{\gamma_{_{SD}}(i)<\Gamma^{\prime}_{th}\}}\right] ,
\end{align}
where $P_{out}^{HU}$ is given by \eqref{huouta}.
\subsection{Direct Transmission}
 In this scheme
 S directly transmits to D. The outage $P_{out}^{DT}$ and throughput $\tau^{DT}$  are given by: $
P_{out}^{DT}=\Pr(\gamma_{_{SD}}(i)<\Gamma_{th})=(1-\exp(-W_{2}\Gamma_{th}))$ and
$\tau^{DT}=R_{0}(1-P_{out}^{DT})=R_{0}\exp(-W_{2}\Gamma_{th})$.
\section{Diversity Analysis for IBEP and IOFP}
In this section we present the diversity order analysis for both IBEP and IOFP. For this reason, we use $\Gamma^{\prime}_{th}=\Gamma_{th}$ throughout this section.
\subsection{IBEP}

To achieve maximum throughput for stable buffers we need to use $\phi_{inc}$ greater than unity (but close to it).  This implies that
$M > \dfrac{\mathbb{E}{\{X(i)}\}}{\left( 1-e^{-W_{2}\Gamma_{th}}\right) }$.  Clearly, when SNR $ \triangleq \dfrac{P_{_{S}}}{\sigma^2} \rightarrow \infty,  W_{2} \rightarrow 0 ,\hspace{0.2 cm} \text{and\ } W_{4} \rightarrow 0.$ Note that when $W_{2} \rightarrow 0$,  $M \rightarrow \infty$, which implies that $ W_{1} \rightarrow 0 \hspace{0.2 cm} \text{and} \hspace{0.2 cm} Z_{1} \rightarrow 0 . $ Although $Z_{1}$ and $M$ are dependent on SNR, their product $Z_{1}M$ is independent of  SNR  ($Z_{1}M=\mathcal{W}(-\phi_{inc}\exp(-\phi_{inc}))+\phi_{inc}$ is constant). This is a crucial observation, and will be used to establish the diversity order of IBEP starting with \eqref{hsuoutag1}.  \\
We first note that the terms $e^{-Z_{1}M}$ and $\left( 1-e^{-Z_{1}M}\right) $ in  \eqref{hsuoutag1} are independent of SNR.
 We know that  $e^{-x} \approx 1-x$ for small $x$,  we get:
\begin{align}
e^{-W_{2}\Gamma_{th}} \approx \left(1-W_{2}\Gamma_{th} \right) , \hspace{0.2 cm} e^{-W_{4}\Gamma_{th}} \approx \left(1-W_{4}\Gamma_{th} \right) , \hspace{0.2 cm} e^{-W_{1}\Gamma_{th}} \approx \left(1-W_{1}\Gamma_{th} \right) .
\label{approximaatsnr}
\end{align}
At high SNR, using \eqref{approximaatsnr},  the second term present  inside the curly brackets of \eqref{hsuoutag1} can be approximated as follows:
\begin{align}
\dfrac{W_{2}\left[ e^{-W_{1}\Gamma_{th}}-e^{-W_{2}\Gamma_{th}}\right] }{(W_{1}-W_{2})} \approx \dfrac{W_{2}\left[ \left( 1-W_{1}\Gamma_{th}\right) -\left( 1-W_{2}\Gamma_{th}\right) \right] }{(W_{1}-W_{2})} \approx -W_2\Gamma_{th}
\label{fistappr}
\end{align}
The product $W_{3}\Gamma_{th}$ is independent of SNR. Furhter, since noise variance is small, the product $W_{3}\Gamma_{th} < 1$ for any given target rate $R_0$. Let $a_{0}=W_3 \Gamma_{th}+\Delta, \text{where $\Delta > 0$ such that $\Delta \rightarrow 0$}$.  Using this, the integral \eqref{integ} can be simplified as follows:
\begin{align}
I_{HSU}&= W_{2}\int_{x=0}^{M}\frac{xZ_{1}e^{-Z_{1}x}[e^{-W_{2}\Gamma_{th}}
	-e^{-\frac{W_{3}\Gamma_{th}}{x}}] }{(xW_{2}-W_{3})}dx \nn \\
&=\underbrace{W_{2}\int_{x=0}^{a_{0}}\frac{xZ_{1}e^{-Z_{1}x}[e^{-W_{2}\Gamma_{th}}
	-e^{-\frac{W_{3}\Gamma_{th}}{x}}] }{(xW_{2}-W_{3})}dx}_{I_{1}}+\underbrace{W_{2}\int_{x=a_{0}}^{M}\frac{xZ_{1}e^{-Z_{1}x}[e^{-W_{2}\Gamma_{th}}
	-e^{-\frac{W_{3}\Gamma_{th}}{x}}] }{(xW_{2}-W_{3})}dx}_{I_{2}} \nn \\
& \overset{t}{\approx }   W_{2}\int_{x=a_{0}}^{M}\frac{xZ_{1}e^{-Z_{1}x}[e^{-W_{2}\Gamma_{th}}
	-e^{-\frac{W_{3}\Gamma_{th}}{x}}] }{(xW_{2}-W_{3})}dx \nn \\
& \approx  W_{2} \int_{x=a_{0}}^{M}\frac{xZ_{1}e^{-Z_{1}x}\left[ \left( 1-W_{2}\Gamma_{th}\right) 
-\left( 1-\frac{W_{3}\Gamma_{th}}{x}\right) \right]  }{(xW_{2}-W_{3})}dx 
=- \Gamma_{th}W_{2} \int_{x=a_{0}}^{M} Z_{1}e^{-Z_{1}x}  \nn \\
&= - \Gamma_{th}W_{2} \left[\left(1-e^{-Z_{1}M} \right)- \left(1-e^{-Z_{1}a_{0}} \right)
 \right] 
 \approx  - \Gamma_{th}W_{2} \left(1-e^{-Z_{1}M} \right).
 \label{approxima}
\end{align}
In the above, approximation $t$ results because $I_{1}$ is  insignificant when SNR is high (this is  due to the fact that the upper limit in $I_{1}$ is $a_{0},$  which is small, and a constant  independent of SNR).  On the other hand, the upper limit in $I_{2}$ is $M,$ which increases with SNR.  The last approximation in \eqref{approxima} is due to the fact that at high SNR: $Z_{1} \rightarrow 0$ and $a_{0} < 1$,  making the term $\left(1-e^{-Z_{1}a_{0}} \right) \rightarrow 0$. Using \eqref{approximaatsnr}, \eqref{fistappr} and \eqref{approxima} the outage probability expression given in \eqref{hsuoutag1}  can be approximated as follows:
\begin{align} 
P_{out}^{HSU}  \approx &\hspace{0.2 cm} W_{4}\Gamma_{th}W_{2}\Gamma_{th}
=K_{0}\left( \frac{1}{\text{SNR}}\right) ^{2} \hspace{0.8cm} \phi_{inc} > 1
\label{divbest}
\end{align} 
where $K_{0}=d_{_{SD}}^\alpha d_{_{SR}}^\alpha {\Gamma_{th}}^2$. \\
 Clearly, the diversity order = $-\lim\limits_{\text{SNR} \rightarrow \infty}\dfrac{\log P_{out}^{HSU} }{\log \text {(SNR)} } =2 $. 
\subsection{IOFP}
Proceedings as with IBEP, the outage probability expression for the IOFP given in \eqref{hsuonoffoutag1}  can be approximated at high SNR as follows:
\begin{align} 
P_{out}^{HSU} \approx & \hspace{0.2 cm} \frac{1}{\phi_{inc}} \left( W_{4}\Gamma_{th}\right)  \left( W_{2}\Gamma_{th}\right) + 
\left( 1-\frac{1}{\phi_{inc}}\right) \left( W_{2}\Gamma_{th}\right)
\hspace{0.4 cm}  \phi_{inc} > 1
\label{outa5}
\end{align}
\begin{align} 
P_{out}^{HSU} \approx & \hspace{0.2 cm} K_{1}\left( \frac{1}{\text{SNR}}\right)^{2}  +K_{2} \left( \frac{1}{\text{SNR}}\right) \hspace{0.4 cm}\phi_{inc} > 1
\label{outa6}
\end{align}
where $K_{1}=\dfrac{{\Gamma_{th}}^{2}d_{_{SR}}^\alpha d_{_{SD}}^\alpha}{\phi_{inc}}$,  and $K_{2}=\left( 1-\dfrac{1}{\phi_{inc}}\right)\Gamma_{th}d_{_{SD}}^\alpha$. Again \eqref{outa6} can be approximated when  $\text{SNR} \rightarrow \infty$ as follows:
\begin{align}
P_{out}^{HSU} \approx & \hspace{0.2 cm} K_{2} \left( \frac{1}{\text{SNR}}\right) \hspace{0.4 cm}\phi_{inc} > 1
\end{align}
Clearly, the diversity order = $-\lim\limits_{\text{SNR} \rightarrow \infty}\dfrac{\log P_{out}^{HSU} }{\log \left(\text{SNR} \right) } =1 $.\\
In the case of IOFP, the diversity order is limited to $1$,  due to the fact that the relay is silent in some signalling intervals  when  $M$ amount of energy is not present in the buffer (note that $M$ is large at high SNRs). As a result, the second term present in  \eqref{outa6} restricts the diversity order to $1$. \\
{\underline{NOTE:} For the case of $\phi_{inc} \leq 1,$     it is apparent from Theorems \ref{Theorem_Limiting_PDF not exists} and  \ref{on-off not exists} that $M$ amount of energy is almost always present in the buffer. Using \eqref{hsuoutage2} it can be easily shown that diversity order is $2$ for both the policies when the buffer is unstable.  

\section{simulation and numerical results}
In this section we present simulations to validate the derived expressions and to obtain insights into performance. In all the simulations, S, R and D are assumed to be located in a two dimensional plane. S is located at (0,0)  and D at (4,0). Further, path-loss exponent $\alpha=3$, and noise power $\sigma^2=-40$ dB. Since performance plots are presented for both incremental and non-incremental relaying, we will find it convenient to use $\phi$ to represent either $\phi_{inc}$ or $\phi_{ninc}$. For the existence of a limiting distribution $\phi=\phi_{inc}=\phi_{ninc}=1.1$ is chosen for all the figures except Fig.~\ref{outagevspinc}.  
  
  \begin{figure}[h!]
  	\begin{minipage}{.5\textwidth}
  		\centering
  		\includegraphics[width=1.2 \linewidth]{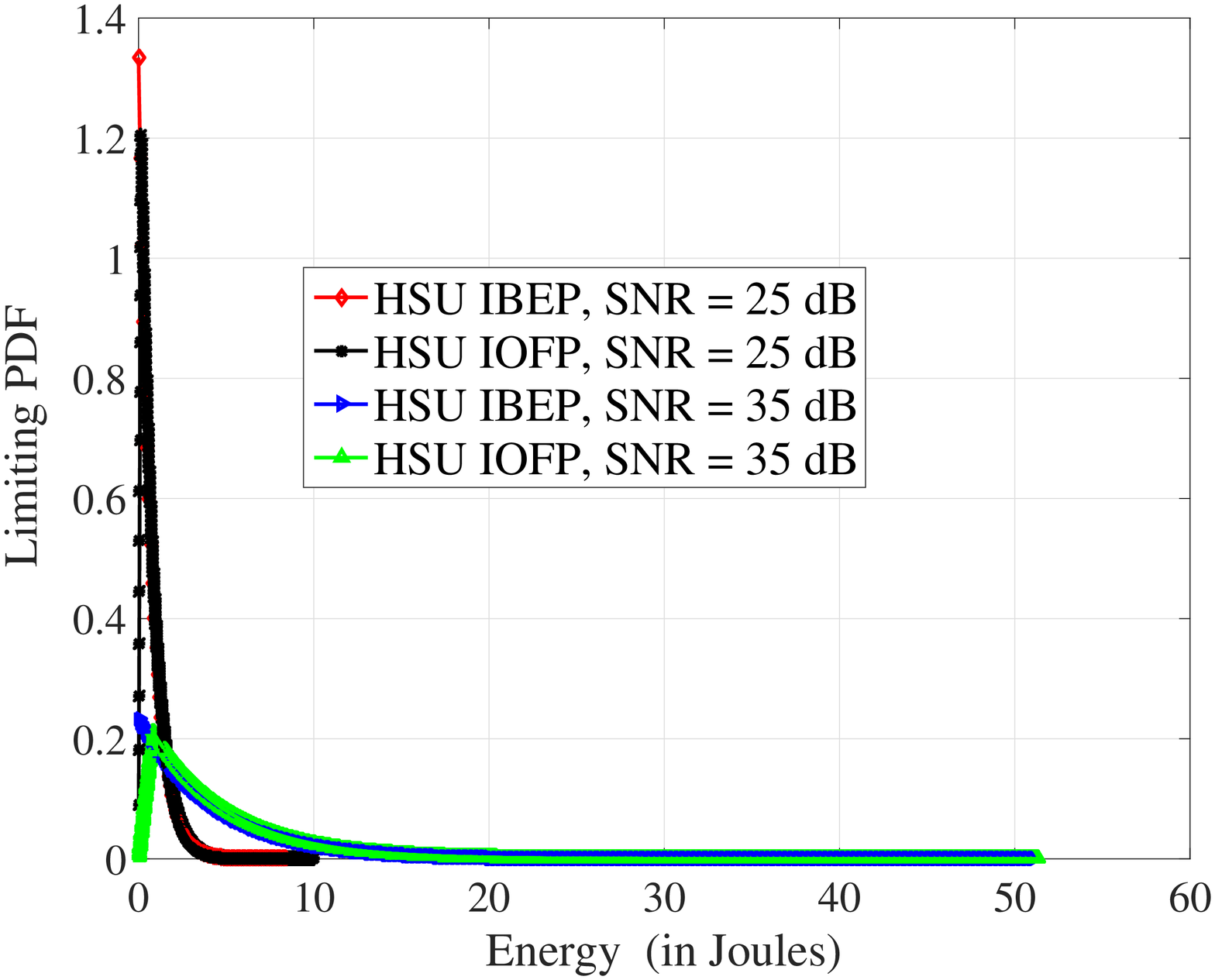}
  		\caption{\small Limiting distribution of energy with  infinite- \\size PEB and parameters $\frac {1}{\lambda_{1}}\hspace{-0.05in}=\hspace{-0.05in}-10$ dB, $\phi_{inc}=1.1$, \\$ R_{0}=1.5$ bpcu,  and relay location $(1,0)$.}
  		\label{limtinpdf}
  	\end{minipage}
  	\begin{minipage}{.5\textwidth}
  		\centering
  		\includegraphics[width=0.9\linewidth]{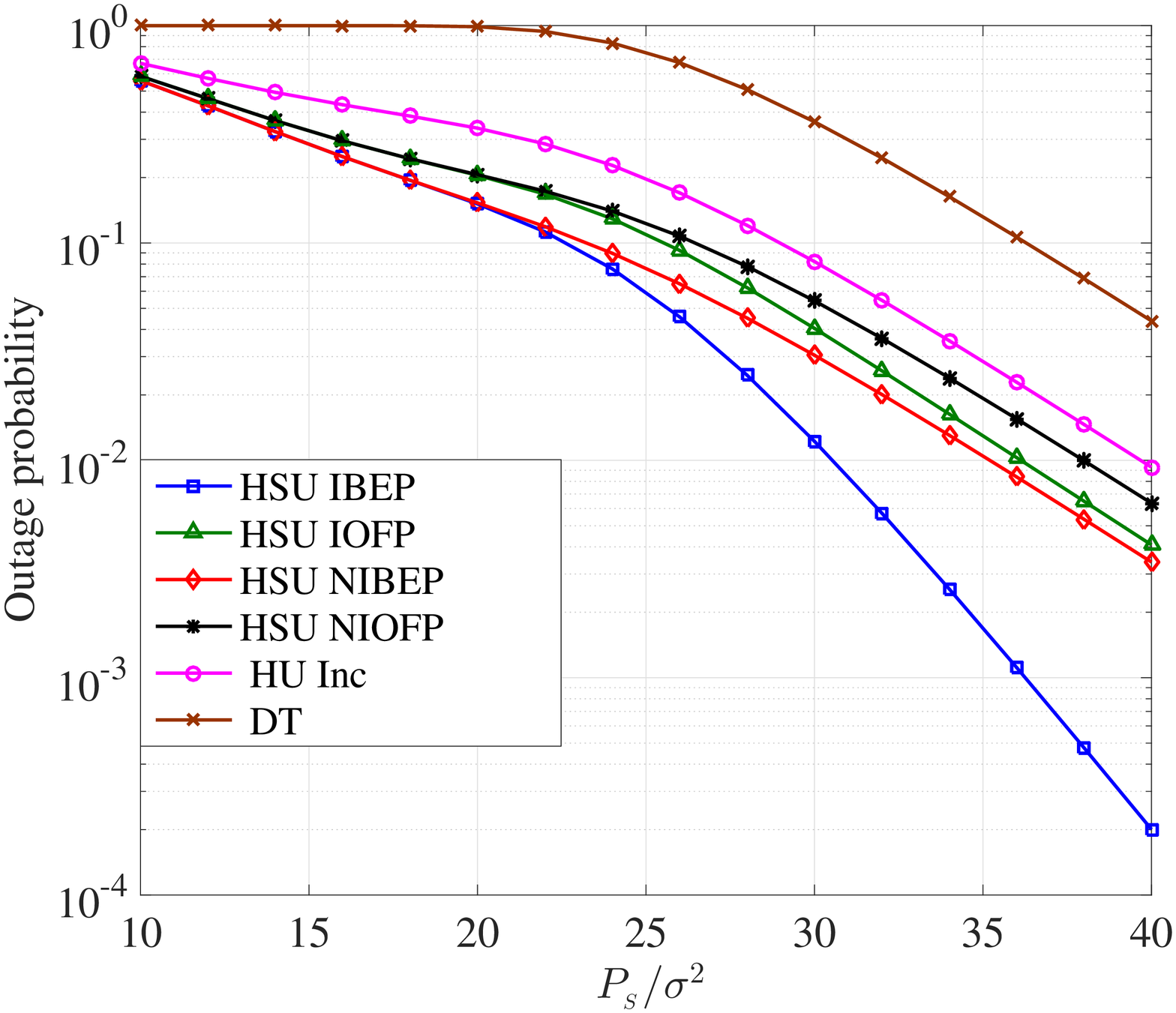}
  		\caption{\small Outage probability vs. source SNR $P_{_{S}}/{\sigma^{2}}$  with  parameters $\frac {1}{\lambda_{1}}=-10$ dB, relay location $(1,0)$,  $\phi=\phi_{inc}=\phi_{ninc}=1.1$,  and target rate $ R_{0}=1.5$ bpcu.}
  			\label{outavssnr}
  	\end{minipage}
  	\vspace*{-0.5cm}
      \end{figure}
  Fig.~\ref{limtinpdf} depicts the limiting distribution of energy with IOFP and IBEP  for two different source SNRs ($25$ dB and $35$ dB) assuming a fixed target rate $R_{0}=1.5$ bits per channel use (bpcu). The  analytical values are obtained using  \eqref{limpdfsplifin} and derivative of \eqref{limcdf}.  In this, and all figures that follow, the solid lines represent analytical plots, and the markers denote simulation values. It is evident from Fig.~\ref{limtinpdf} that when  the source SNR increases, the tail of the limiting PDF increases. This is because at higher SNR, R needs to transmit in fewer time slots due to frequent ACK from D (R therefore accumulates more energy).  
  \par Fig.~\ref{outavssnr} depicts the variation  of outage probability with  source SNR. If the transmit power of the source increases, the outage probability decreases for all three schemes (HSU, 
  HU and DT) due to the strong direct link between S and D.  As compared to DT and HU schemes, HSU yields better performance, clearly indicating that use of energy buffers is advantageous. One reason for this is that use of PEB with SEB in HSU allows R to charge and accumulate energy even when it is transmitting. Another reason is that at high source SNRs there is  accumulation of energy at R due to frequent ACK from D. This is not possible with the HU scheme. From Fig.~\ref{outavssnr} it can also be observed that at  low SNRs performance of incremental and non-incremental schemes with best-effort and on-off  policies is quite similar . But when the SNR is higher, IBEF is clearly superior in performance to all the other schemes, and  achieves full diversity of $2$ even with stable buffers. IOFP achieves lower diversity since R is silent  (because $B(i)<M$) in larger number of signalling intervals.   }
  \begin{figure}[h!]
  	\begin{minipage}{.3\textwidth}
  		\centering
  		\includegraphics[width=1 \linewidth]{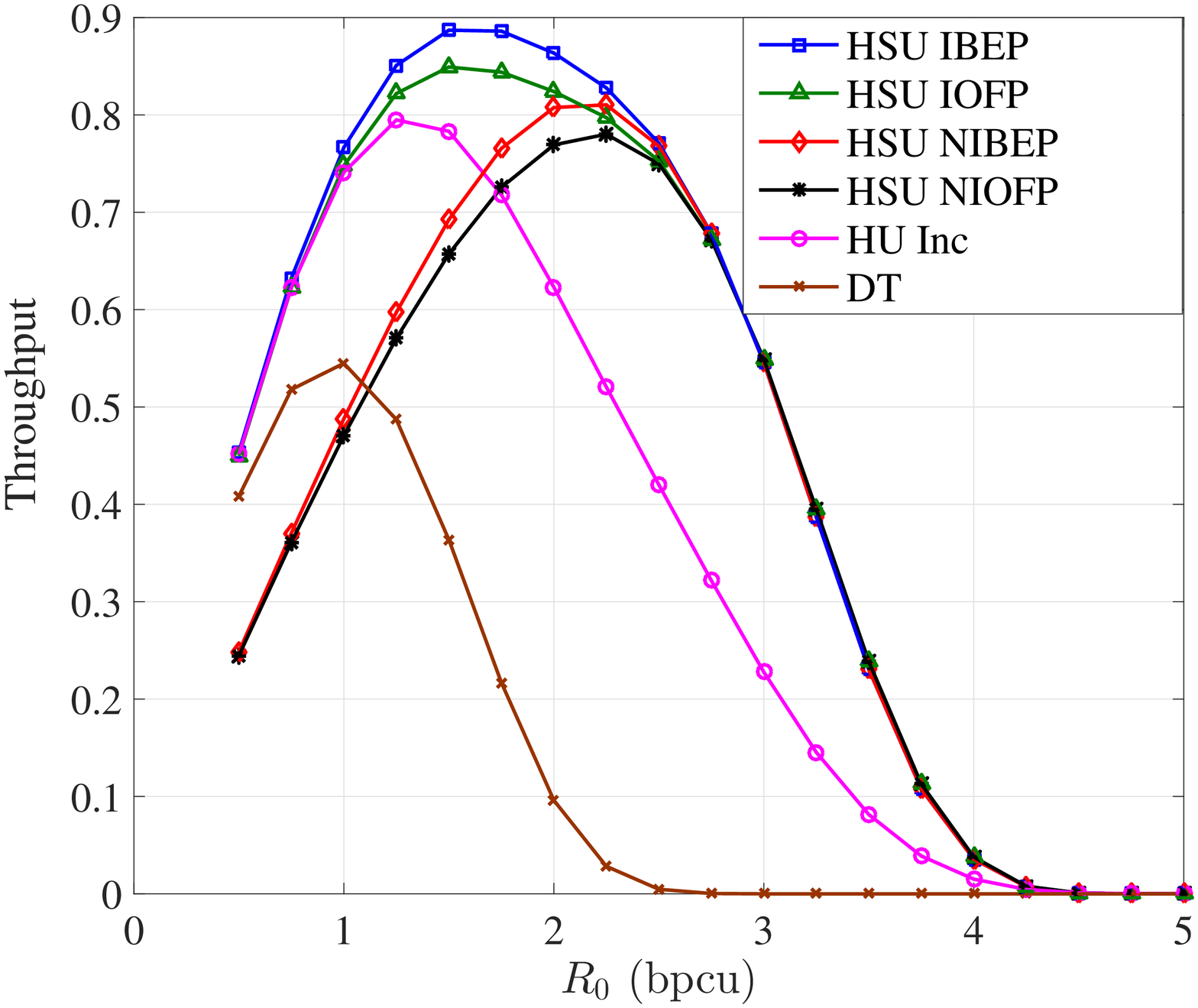}
  		\caption{\small Throughput vs. target rate $ R_{0}$ with parameters source SNR $P_{_{S}}/{\sigma^{2}}=25$ dB,  $\frac {1}{\lambda_{1}}=-10$ dB, \\ $\phi=\phi_{inc}=\phi_{ninc}=1.1$, and relay location $(1,0)$.}
  		\label{thrptvstarge}
  	\end{minipage}~
  	\begin{minipage}{.3\textwidth}
  		\centering
  		\includegraphics[width=1 \linewidth]{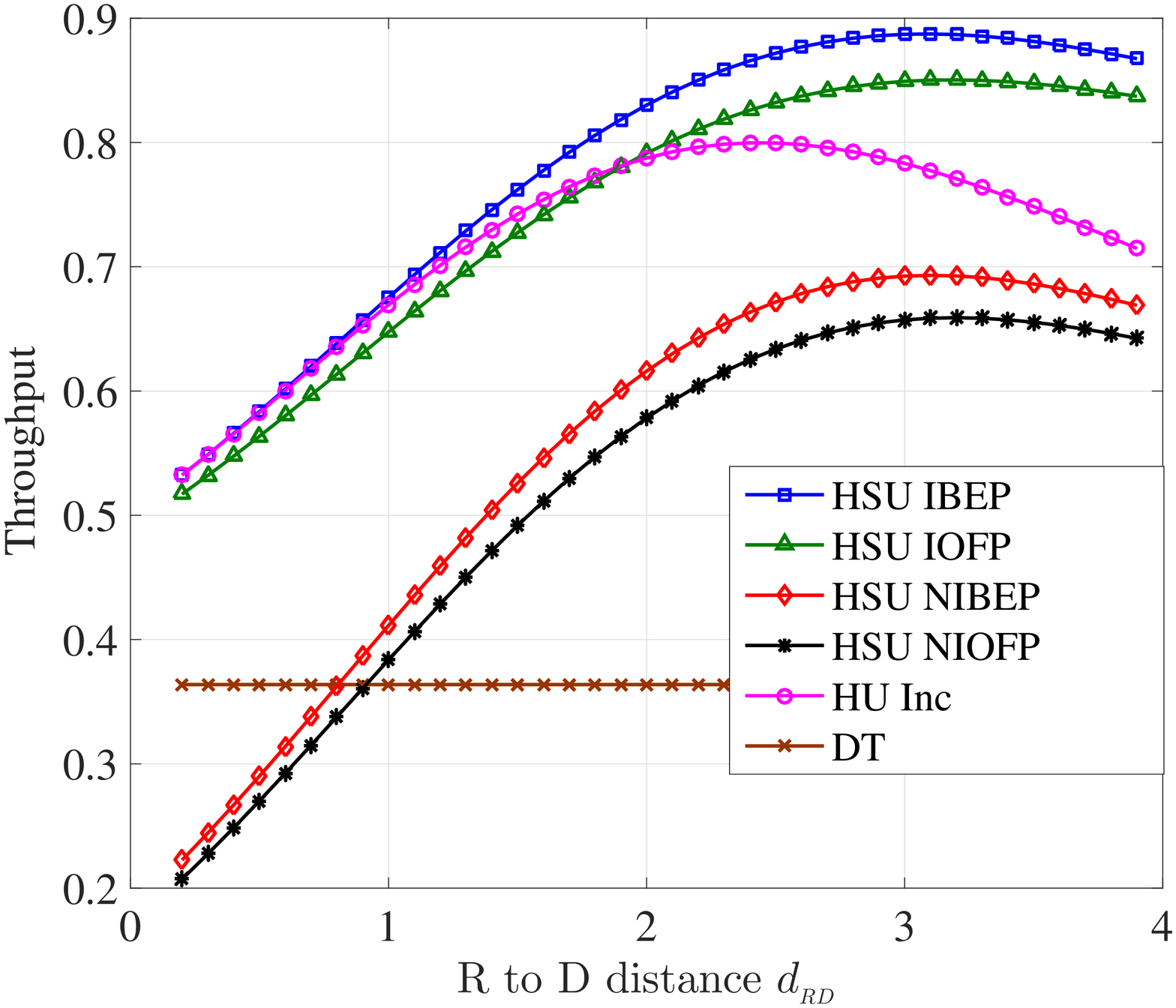}
  		\caption{\small Throughput vs. R to D distance  with parameters source SNR $P_{_{S}}/{\sigma^{2}}=25$ dB, $\frac {1}{\lambda_{1}}=-10$ dB, $\phi=\phi_{inc}=\phi_{ninc}=1.1$  and target rate $R_{0} = 1.5$ bpcu.}
  		\label{thrptvsrtod}
  	\end{minipage}~
  \begin{minipage}{.3\textwidth}
  	\centering
  	\includegraphics[width=1 \linewidth]{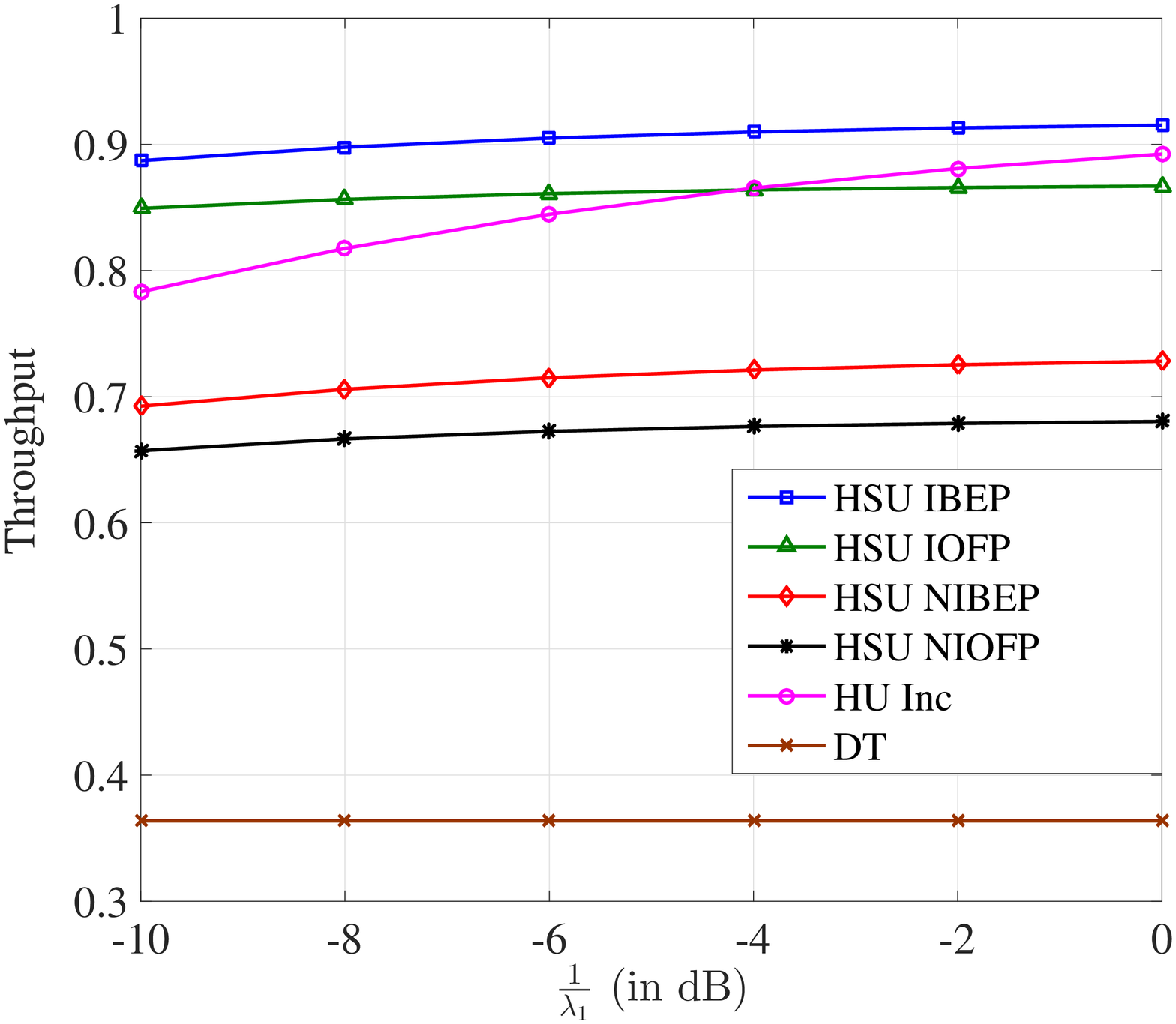}
  	\caption{\small Throughput vs. mean $ \frac{1}{\lambda_{1}} $ with source SNR  $P_{_{S}}/{\sigma^{2}}=25$ dB, relay location $(1,0)$, $\phi=\phi_{inc}=\phi_{ninc}=1.1$ and $ R_{0}= 1.5$ bpcu.}
  	\label{thrptvsmean}
  \end{minipage}~
  \vspace*{-0.5 cm}
  \end{figure}
\par Fig.~\ref{thrptvstarge} depicts the variation of throughput  with target rate $R_{0}$. From Fig.~\ref{thrptvstarge} we observe that at very low target rates $R_{0}$, HSU and HU yield practically the same throughput (which is quite intuitive).  If the target rate $R_{0}$ is increased, HSU performs better than HU. Further increase in $R_{0}$ results in lower throughput performance. It can be seen that HSU schemes result in larger throughput as compared to HU and DT, and can support larger target rates $R_0$ (due to use of the energy buffer). From Fig.~\ref{thrptvstarge}  it can also be observed that for low-to-medium target rates,  the performance of incremental schemes is superior to that of non-incremental schemes. For higher target rates, the performance is quite similar due to increase in frequency of NACK from D.
\par Fig.~\ref{thrptvsrtod} depicts the variation of throughput with respect to R-D distance  $d_{_{RD}}$ keeping S-D distance $d_{_{SD}}$ the same $\left(d_{_{SD}}=d_{_{SR}}+d_{_{RD}} \right) $. From Fig.~\ref{thrptvsrtod} it can be observed that when the R is located away from D, throughput increases for both HSU and HU upto some point since the probability of successful decoding at R increases (path-loss $d_{_{SR}}^{\alpha}$ decreases). However, as R moves closer to S, the increasing R-D distance causes loss in throughput as it reduces the probability of successful decoding at D (path-loss $d_{_{RD}}^{	\alpha}$ increases). Thus there exists an optimal R location for both HSU and HU as can be observed from Fig.~\ref{thrptvsrtod}. Note also that HSU results in higher throughputs as compared to HU when R is located further away from the D  (this once again is due to the energy buffer). From Fig.~\ref{thrptvsrtod} it can also be observed that the performance of incremental schemes is always superior to that of the non-incremental schemes for a given  $R_{0}$.\\ 
\begin{figure}[t!]
	\begin{minipage}{.5\textwidth}
		\centering
		\includegraphics[width=0.95\linewidth]{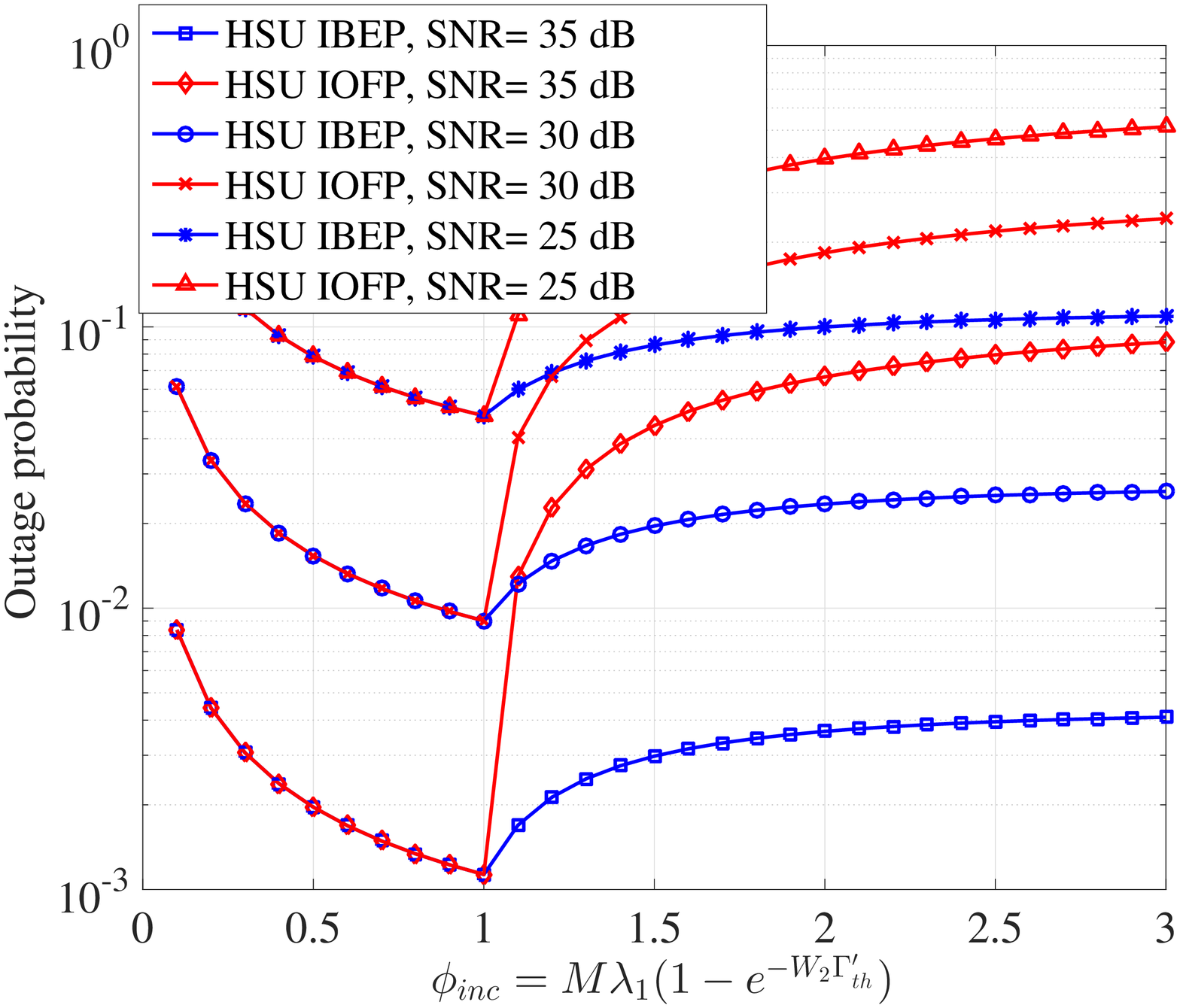}
		\caption{\small Outage probability vs. $\phi_{inc}$ with \\ parameters  $\frac {1}{\lambda_{1}}=-10$ dB, relay location $(1,0)$ and \\
			target rate  $ R_{0} = 1.5$ bpcu.}
		\label{outagevspinc}
	\end{minipage}
	\begin{minipage}{.5\textwidth}
	\centering
	\includegraphics[width=0.95\linewidth]{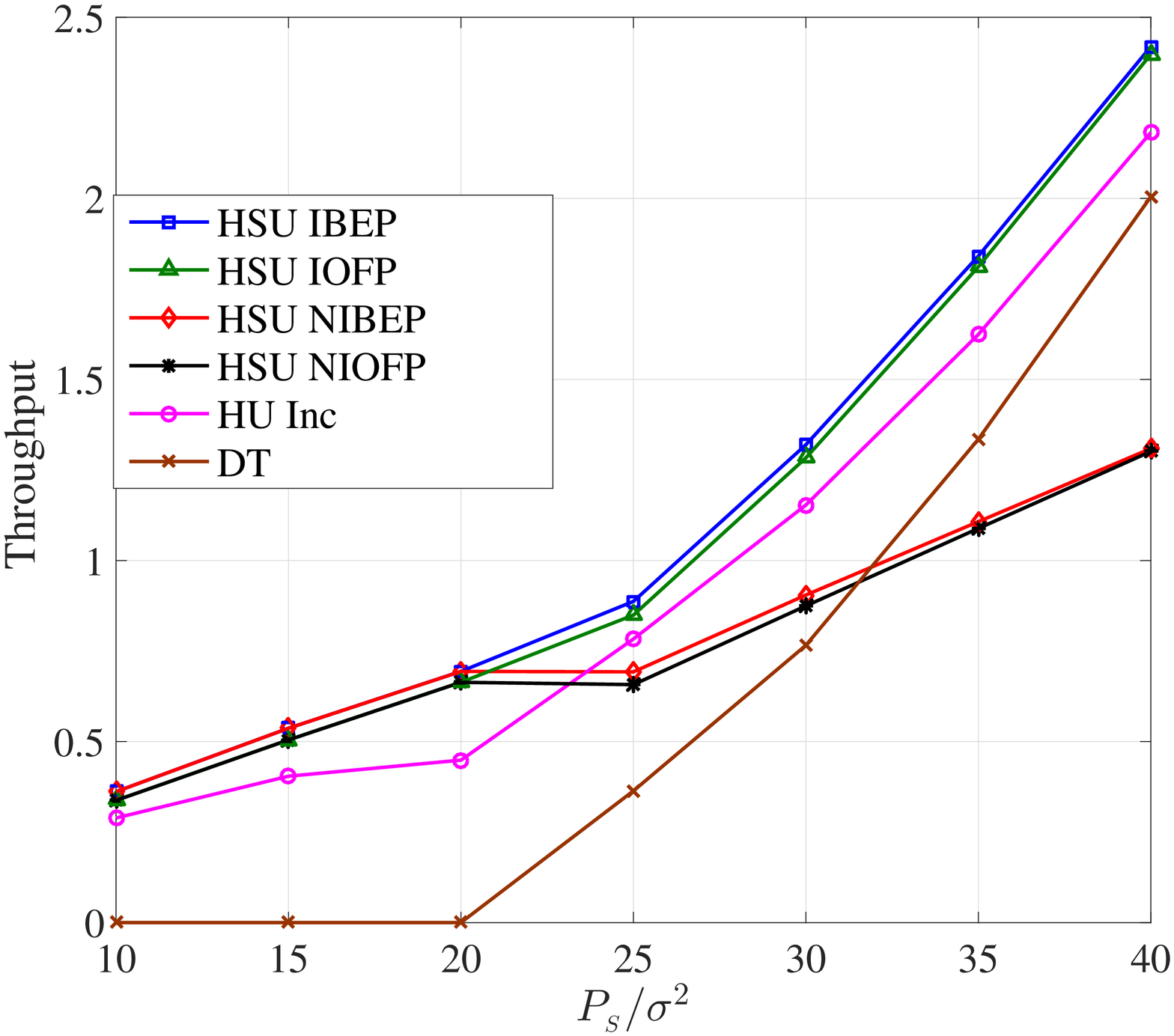}
	\caption{\small Throughput vs. source SNR $P_{_{S}}/{\sigma^{2}}$ with \\ parameters  $\frac {1}{\lambda_{1}}=-10$ dB, $\phi=\phi_{inc}=\phi_{ninc}=1.1$, \\ relay location $(1,0)$, and optimum target rate $ R_{0}$.}
	\label{thrptvssnr}
\end{minipage}
\vspace*{-1 cm}
\end{figure}
Fig.~\ref{thrptvsmean} depicts the variation of throughput with respect to mean of the incoming EH process  $\frac{1}{\lambda_{1}}$ (in dB). Clearly, throughput increases with increase in $ \frac{1}{\lambda_{1}}$ for HSU and HU.   With HSU architecture, as the mean $(\frac{1}{\lambda_{1}})$  increases, the energy stored at R increases. However,  the energy drawn from the buffer ($M$) also simultaneously  increases for both incremental and non-incremental schemes (since $\phi=\phi_{ninc}=\phi_{inc}$ is a constant). For large $ \frac{1}{\lambda_{1}}$ values, performance of HU is closer to that of HSU with on-off policy. From Fig.~\ref{thrptvsmean} it can also be observed that the incremental schemes result in performance that is superior to that of non-incremental schemes.  
\par Fig.~\ref{outagevspinc} shows that variation of outage vs. $\phi_{inc} $ for three different source SNRs ($25$, $30$ and $35$ dB). When $\phi_{inc}<1$, performance of IBEP and IOFP is the same (this is because $M$ units of energy is available almost surely with both schemes). When $\phi_{inc}>1$, IBEP results in better performance (the gap is particularly pronounced at high SNRs). This is because $M$ increases with increase in $\phi_{inc}$, which causes IOFP to keep R slient in more signalling intervals due to paucity of energy. Note that the optimum $\phi_{inc}$ is close to $1$. From Fig.~\ref{outagevspinc} it can also be observed that for high source SNR ($35$ dB), the performance of IOFP is very poor for $\phi_{inc} > 1$ compared to that of IBEP. This is because a very high $M$ is required at these SNRs, and the relay with IOFP has insufficient energy in some signalling intervals. With IBEP, a small amount of energy ($B(i)<M$) often sufices to ensure sufficient SNR at D.

Fig.~\ref{thrptvssnr} depicts the variation of throughput vs. source SNR $P_{_{S}}/{\sigma^{2}}$ assuming the optimum target rate $R_{0}$ is used at each source SNR. For low SNRs ($10$ to $20$ dB)   performance of HSU incremental and HSU non-incremental schemes is quite similar due to frequent failure of the S-D direct link. But when the SNR increases further (beyond $20$ dB) for the optimum target rate $R_{0}$, there is increase in ACKs with incremental schemes, which enable them to yield higher throughput. From Fig.~\ref{thrptvssnr} it can also be observed that at very high SNR, the throughput gap between incremental and non-incremental is very high for optimum target rate $R_{0}$ due to increase in frequency of ACKs from D with incremental signalling,  which causes energy accumulation in the buffer (this energy is used for relaying in the signalling intervals in which NACK is received from D).  At higher SNR, DT and HU incremental are also better than non-incremental transmission schemes.
 \section{conclusion}
In this paper, performance of energy buffer-aided incremental relaying is analyzed for fixed-rate transmission. A discrete-time continuous state-space Markov chain model is used for the energy buffer, and limiting distributions are derived for the stored energy with incremental best-effort and on-off policies. Expressions are derived for outage probabilities and throughput. Corresponding expressions for non-incremental schemes follow as special cases.  It is shown that with IBEP, diversity of two is attained with stable buffers, while diversity of only one is attained with IOFP. Throughput with IBEP is marginally superior to that with IOFP. 
\appendices
\section {Proof of Theorem \ref{Theorem_Limiting_PDF}} \label{Appendix_Limiting_PDF}
Using the  total probability theorem, the CDF of $B\left( i+1\right) $ can be evaluated for energy storage model of \eqref{BI2} as follows:
\begin{align}
&\Pr \{B\left( i+1\right) \leq x \}   =  \Pr \left\lbrace  \left( B\left( i\right)+X\left( i\right) \right)  \leq x \mid \gamma_{_{SD}}(i) \geq \Gamma^{\prime}_{th}\right\rbrace 
\Pr \{\gamma_{_{SD}}(i) \geq \Gamma^{\prime}_{th}\}\nn \\
&+ 
\Pr \left\lbrace (\gamma_{_{SD}}(i)<\Gamma^{\prime}_{th}), (B(i)\geq M)\right\rbrace 
\Pr \left\lbrace  \left( B\left( i\right)-M+X\left( i\right)\right)   \leq x \mid  (\gamma_{_{SD}}(i)<\Gamma^{\prime}_{th}), (B(i)\geq M) \right\rbrace \nn \\
&+ \Pr \{(\gamma_{_{SD}}(i)<\Gamma^{\prime}_{th}), (B(i)<M)\} 
\Pr \{ X\left( i\right) \leq x \mid  (\gamma_{_{SD}}(i)<\Gamma^{\prime}_{th}), (B(i)<M) \}.
\label{buffeq3}
\end{align}
 Let $B(i)=U$, $X(i)=X $ and $\gamma_{_{SD}}(i)=\gamma_{_{SD}}$.  Then  \eqref{buffeq3} can be written in terms of joint probabilities as:
 \begin{align}
 \Pr \{B\left( i+1\right) \leq x \}= & \Pr \{ U+X \leq x , (\gamma_{_{SD}}\geq \Gamma^{\prime}_{th}) \}\nn \\
 &+\Pr \{ U-M+X \leq x , (\gamma_{_{SD}}<\Gamma^{\prime}_{th}), (U\geq M) \} \nn \\ 
 & +\Pr \{ X \leq x , (\gamma_{_{SD}}<\Gamma^{\prime}_{th}), (U<M) \}.
 \label{10}
 \end{align}
 Let $G_{i}(x)$ be the CDF of $ B(i)$. Then $G_{i+1}(x)$ is the CDF of $B(i+1)$ i.e. $\Pr \{B\left( i+1\right) \leq x \}=G_{i+1}(x)$.
 After a long time ($i\rightarrow\infty$), the buffer reaches its steady-state so that $G_{i}(x)=G_{i+1}(x)=G(x).$ 
 In this state, \eqref{10} can be written as follows:
 \begin{align}
 G(x)= &\int_{y=\Gamma^{\prime}_{th}}^{\infty}\int_{u=0}^{x}F_{X}\left( x-u\right)g\left( u\right)f_{\gamma_{_{SD}}}\left( y\right)  du dy \nn \\
&+ \int_{y=0}^{\Gamma^{\prime}_{th}}\int_{u=M}^{M+x}F_{X}\left( x-u+M\right)g\left( u\right)f_{\gamma_{_{SD}}}\left( y\right)  du dy \nn \\
 &+\int_{y=0}^{\Gamma^{\prime}_{th}}\int_{u=0}^{M}F_{X}\left( x\right)g\left( u\right)f_{\gamma_{_{SD}}}\left( y\right)  du dy,
 \label{11}
 \end{align}
 where $g(x)$ is the derivative of  $G(x)$. 
 The PDFs of $\gamma_{_{SD}}$ and $X$ are given by: $ f_{\gamma_{_{SD}}}\left( y\right) =W_{2}\exp\left( -W_{2}y\right)$ and  $F_{X}\left( x\right) =1-\exp\left( -\lambda_{1}x\right)$.
  Substituting $ f_{\gamma_{_{SD}}}\left( y\right) $ in \eqref{11} we get:
  \begin{align}
  G(x)= \hspace{0.1 cm}&e^{-W_{2}\Gamma^{\prime}_{th}}\int_{u=0}^{x}F_{X}\left( x-u\right)g\left( u\right)du \nn \\
 & +\left( 1-e^{-W_{2}\Gamma^{\prime}_{th}}\right)\left[\int_{u=M}^{M+x}F_{X}\left( x-u+M\right)g\left( u\right)du + 
 \int_{u=0}^{M}F_{X}\left( x\right)g\left(u\right)du \right]  .
  \label{limcdfints}
  \end{align}
  Consider the term $ \displaystyle \int_{u=M}^{M+x}F_{X}\left( x-u+M\right)g\left( u\right)du+
\int_{u=0}^{M}F_{X}\left( x\right)g\left(u\right)du $. 
  Integrating by parts and using $ G\left( 0\right) =0$, the term can be simplified as follows:
 \begin{align*}
 =\hspace{0.1 cm}& F_{X}\left( x-u+M\right)G(u)\Bigm|_{u=M}^{M+x}+ 
 \int_{u=M}^{M+x}f_{X}\left(x-u+M\right) G\left( u\right) du+F_{X}\left( x\right) G\left( M\right) \\
= & -F_{X}\left( x\right) G\left( M\right)+\int_{u=M}^{M+x}f_{X}\left(x-u+M\right) G\left( u\right) du+F_{X}\left( x\right) G\left( M\right) \\
 =& \int_{u=M}^{M+x}f_{X}\left(x-u+M\right) G\left( u\right) du.
  \end{align*}
Similarly,  $\displaystyle\int_{u=0}^{x}F_{X}\left( x-u\right) g\left( u\right)du\hspace{-0.02in}=\hspace{-0.02in}\int_{u=0}^{x}G\left( u\right) f_{X}\left( x-u\right) du$ so that
 \eqref{limcdfints} can be simplified as: 
\begin{align}
 G(x)= \hspace{0.1 cm} & e^{-W_{2}\Gamma^{\prime}_{th}}\displaystyle\int\limits_{u=0}^{x}f_{X}\left( x-u\right)G\left( u\right)du 
 +\left( 1-e^{-W_{2}\Gamma^{\prime}_{th}}\right)\int\limits_{u=M}^{M+x}f_{X}\left( x-u+M\right)G\left( u\right)du.
 \label{13}
 \end{align}
 Using  $v=x-u+M$,  the second integral in \eqref{13} can be simplified as: \\ $\displaystyle\int_{u=M}^{M+x}f_{X}\left( x-u+M\right)G\left( u\right)du=\int_{u=0}^{x}f_{X}\left(u\right)G\left( x+M-u\right)du.$\\
With this, \eqref{13} can be simplified and rewritten as in \eqref{limcdfintegr}.
The integral equation \eqref{limcdfintegr}  can be solved by postulating a solution of the form \cite{D.V.Lindley1952},\cite{J.Gani1957}: 
\begin{align}
G\left( x\right) =\sum_{i}c_{i}e^{ -Z_{i}x}.
\label{cdfl}
\end{align} 
Substituting $f_{X}\left( x\right) =\lambda_{1}e^{ -\lambda_{1}x}$ in \eqref{limcdfintegr} and using \eqref{cdfl} we get:
\begin{align}
\sum_{i}c_{i}e^{-Z_{i}x}=\hspace{0.1 cm}& e^{-W_{2}\Gamma^{\prime}_{th}}\lambda_{1}e^{ -\lambda_{1}x} \int_{u=0}^{x}e^{ \lambda_{1}u} \sum_{i}c_{i}e^{-Z_{i}u}du \nn \\
  &+\left( 1-e^{-W_{2}\Gamma^{\prime}_{th}}\right)\lambda_{1}\int_{u=0}^{x}e^{-\lambda_{1}u} \sum_{i}c_{i}e^{ -Z_{i}\left(x+M-u\right)}du.
\end{align}
Exchanging the summation and integration, and  simplifying,  we get:
\begin{align}
\sum_{i}c_{i}e^{-Z_{i}x}= &\sum_{i}\left[  \frac {e^{-W_{2}\Gamma^{\prime}_{th}}\lambda_{1}+\left( 1-e^{-W_{2}\Gamma^{\prime}_{th}}\right)\lambda_{1}e^{ -Z_{i}M} }{\left( \lambda_{1}-Z_{i}\right)} \right]c_{i}e^{-Z_{i}x} \nn \\
&-e^{-\lambda_{1}x} 
\sum_{i}\left[  \frac {e^{-W_{2}\Gamma^{\prime}_{th}}\lambda_{1}+\left( 1-e^{-W_{2}\Gamma^{\prime}_{th}}\right)\lambda_{1}e^{-Z_{i}M} }{\left( \lambda_{1}-Z_{i}\right)} \right]c_{i}. 
\end{align}
For both sides of the above equation to be equal, the following conditions need to be satisfied:
\begin{align}
\frac {e^{-W_{2}\Gamma^{\prime}_{th}}\lambda_{1}+\left( 1-e^{-W_{2}\Gamma^{\prime}_{th}}\right)\lambda_{1}e^{-Z_{i}M} }{\left( \lambda_{1}-Z_{i}\right)}=1  \hspace{0.5cm} \forall i ,\hspace{1 cm} \text{and} \hspace{1 cm} \sum_{i}c_{i}=0.
\label{17}
\end{align}
Clearly, \eqref{17} has two roots \cite{J.Gani1957}, one default  root is $Z_{0}=0$ for $i=0$, and other root  $Z_{1}$ for $i=1$ can be seen to satisfy the following equation:
\begin{align}
e^{-W_{2}\Gamma^{\prime}_{th}}\lambda_{1}+\left( 1-e^{-W_{2}\Gamma^{\prime}_{th}}\right)\lambda_{1}e^{-Z_{1}M}=\left( \lambda_{1}-Z_{1}\right),\\ \Rightarrow
e^{ -Z_{1}M} =\frac{-Z_{1}M}{\left( 1-e^{-W_{2}\Gamma^{\prime}_{th}}\right)\lambda_{1}M}+1.
\label{othroot}
\end{align}
With $b=1$, $a=\dfrac{1}{\left( 1-e^{-W_{2}\Gamma^{\prime}_{th}}\right)\lambda_{1}M}$, $p=e$,  $x=-Z_{1}M$,
the above is of the form $p^{x}=ax+b$,
 whose solution is $ x=-\mathcal{W}\left( -a^{-1}p^{-\frac{b}{a}}\ln p\right)-\dfrac{b}{a} $,  when $ p > 0 $ and $ p\neq 1$ ($ a \neq 0$),
where $\mathcal{W}\left( \cdot\right) $ is Lambert-W function \cite{Tony}. This implies that:
\begin{align}
Z_{1}=\frac{\mathcal{W}\left( -\frac{1}{a}\exp{\left( -\frac{1}{a}\right) }\right)}{M}+\frac{1}{aM}.
\label{z1vas}
\end{align}
For the limiting distribution to exist,  $Z_{1}$ should be positive $(Z_{1} >0)$ so that
\begin{align}
\hspace{1 cm} & \frac{\mathcal{W}\left( -\frac{1}{a}\exp{\left( -\frac{1}{a}\right) }\right)}{M}+\frac{1}{aM} >0
\label{22}
\end{align}
Using the property of Lambert-W function
$\mathcal{W}\left( -x\exp\left( -x\right)\right) =-x \hspace{0.1 cm}$ ($0< x\leq 1$), 
the condition for existence of limiting distribution can be seen to be: 
\begin{align}
\frac{1}{a} > 1 &\Rightarrow \frac{1}{\left( 1-e^{-W_{2}\Gamma^{\prime}_{th}}\right)\lambda_{1}M} <1 \Rightarrow M>\frac{1}{\left( 1-e^{-W_{2}\Gamma^{\prime}_{th}}\right)\lambda_{1}} \nn \\
& \Rightarrow M > \frac {\mathbb{E}{\{X(i)}\}}{\left( 1-e^{-W_{2}\Gamma^{\prime}_{th}} \right) }.
\end{align}
It is clear from (\ref{phinc}) that the above implies that $\phi_{inc} >1$. From \eqref{cdfl}, we have:
\begin{align}
G\left( x\right) =\sum_{i}c_{i}e^{-Z_{i}x}=c_{0}e^{-Z_{0}x} +c_{1}e^{ -Z_{1}x}=c_{0}+c_{1}e^{-Z_{1}x} .
\label{cf}
\end{align}
Since $G\left( x\right) $ is limiting CDF of energy stored in the PEB, $G(0)=0$ and $G(\infty)=1$. The latter implies that $c_0=1$ and the former shows that $c_0+c_1=0$ ($c_1=-1$).
 Substituting the values of $c_{0}$,  $c_{1}$ in \eqref{cf}, we get
$ G\left( x\right) =1-\exp\left( -Z_{1}x\right)$. After substituting  for `$a$'  in \eqref{z1vas} we obtain the  $Z_{1}$ as in \eqref{limcdf}.
\section{Proof of Theorem \ref{Theorem_Limiting_PDF not exists}} \label{Appendix_Limiting_PDF not exists}
Let the S-D link outage indicator variable $ \mathbb{O}_{_{SD}}\left( i\right) \in \left( 0,1\right) $  be  defined as follows:
\begin{equation}
\mathbb{O}_{_{SD}}\left( i\right)\triangleq\begin{cases}
1, & \text{if $\gamma_{_{SD}}\left( i\right) < \Gamma^{\prime}_{th}$}\\
0, & \text{if $\gamma_{_{SD}}\left( i\right) \geq \Gamma^{\prime}_{th}$}.
\end{cases}
\label{indi}
\end{equation}
Let $E_{R}\left( i\right)= \min\left( B\left( i\right),M\right).$  With \eqref{indi}, the buffer update equations in \eqref{BI2} can be rewritten as:
\begin{align}
B\left( i+1\right) - B\left( i\right)  = X\left( i\right)-E_{R}\left( i\right)\mathbb{O}_{_{SD}}\left( i\right).
\label{buffequal}
\end{align}
Ergodicity implies that
$\lim\limits_{N \rightarrow \infty}\dfrac{1}{N}\sum_{i=1}^{N}\mathbb{O}_{_{SD}}\left( i\right)=\mathbb{E}\{\mathbb{O}_{_{SD}}\left( i\right)\}$ so that:

\begin{align}
\lim\limits_{N \rightarrow \infty}\frac{1}{N}\sum_{i=1}^{N}\mathbb{O}_{_{SD}}\left( i\right)= & \mathbb{E}\{\mathbb{O}_{_{SD}}\left( i\right)\}=  1\times\Pr\{\mathbb{O}_{_{SD}}\left( i\right)=1\}+0\times\Pr\{\mathbb{O}_{_{SD}}\left( i\right)=0\}  \nn \\
=&\Pr\{\mathbb{O}_{_{SD}}\left( i\right)=1\}=\Pr\{\gamma_{_{SD}}\left( i\right) < \Gamma^{\prime}_{th}\}=\left( 1-e^{-W_{2}\Gamma^{\prime}_{th}}\right). 
\label{33}
\end{align}
 From the law of conservation of flow, average arrival rate is always greater than or equal to average departure rate i.e. $\mathcal{A} \geq \mathcal{D}$. $\mathcal{A}=\mathcal{D}$ holds only when the buffer is non-absorbing. 
 Using law of large numbers, the average arrival rate is given by: 
 \begin{align}
 \mathcal{A}=\mathbb{E}\left\lbrace X\left( i\right) \right\rbrace =\lim\limits_{N \rightarrow \infty}\dfrac{1}{N} \sum_{i=1}^{N}X\left( i\right)
 \label{avgarrival}
 \end{align} 
Similarly the average departure rate $\mathcal{D}$ is given by:
\begin{align}
\mathcal{D}=\lim\limits_{N \rightarrow \infty}\frac{1}{N}\sum_{i=1}^{N}E_{R}\left( i\right)\mathbb{O}_{_{SD}}\left( i\right)=\lim\limits_{N \rightarrow \infty}\frac{1}{N}\sum_{B(i)\geq M}M\mathbb{O}_{_{SD}}\left( i\right) +\lim\limits_{N \rightarrow \infty}\frac{1}{N}\sum_{B(i)< M}B(i)\mathbb{O}_{_{SD}}\left( i\right)
\end{align}
 The relay transmit energy $ E_{R}\left( i\right)$ is upper bounded by  $ E_{R}\left( i\right) \leq M \Rightarrow $
\begin{align}
\lim\limits_{N \rightarrow \infty}\frac{1}{N}\sum_{i=1}^{N}E_{_{R}}\left( i\right)\mathbb{O}_{_{SD}}\left( i\right) \leq \lim\limits_{N \rightarrow \infty} \frac{1}{N}\sum_{i=1}^{N}M\mathbb{O}_{_{SD}}\left( i\right)
\label{upperbou}
\end{align} 
For $\phi_{inc} < 1 \Rightarrow $  it follows from \eqref{33} that:
\begin{align}
\mathbb{E}\{X\left( i\right)\} > M\left( 1-e^{-W_{2}\Gamma^{\prime}_{th}}\right)\Rightarrow \lim\limits_{N \rightarrow \infty}  \frac{1}{N}\sum_{i=1}^{N}X\left( i\right) > \lim\limits_{N \rightarrow \infty}\frac{1}{N}\sum_{i=1}^{N}M\mathbb{O}_{_{SD}}\left( i\right) 
\label{inequ} 
\end{align}
From \eqref{upperbou} and \eqref{inequ} we get :
\begin{align}
\lim\limits_{N \rightarrow \infty}\frac{1}{N}\sum_{i=1}^{N}X\left( i\right) > \lim\limits_{N \rightarrow \infty}\frac{1}{N}\sum_{i=1}^{N}E_{_{R} }\left( i\right)\mathbb{O}_{_{SD}}\left( i\right)
\label{buffabsor}
\end{align}
From \eqref{buffequal} and \eqref{buffabsor} we get:
\begin{align}
\lim\limits_{N \rightarrow \infty}\frac{1}{N}\sum_{i=1}^{N}X\left( i\right) > \lim\limits_{N \rightarrow \infty}\frac{1}{N}\sum_{i=1}^{N}E_{_{R} }\left( i\right)\mathbb{O}_{_{SD}}\left( i\right) \Rightarrow \lim\limits_{N \rightarrow \infty}\frac{1}{N}\sum_{i=1}^{N}B\left( i+1\right)>\lim\limits_{N \rightarrow \infty} \frac{1}{N}\sum_{i=1}^{N}B\left( i\right)
\label{expe}
\end{align}
Clearly, \eqref{expe} indicates that the buffer is in an absorbing state ( i.e. $ \mathcal{A} > \mathcal{D}$), which makes it unstable. In this case the average buffer length increases with each signalling interval. After a certain number of signalling intervals, $M$ amount of energy is present in the buffer almost surely i.e. $\Pr \left( B\left( i\right)  \geq M \right) \rightarrow 1. $\\
 For $\phi_{inc}=1 \Rightarrow $ it follows from \eqref{33} that:
 \begin{align}
\mathbb{E}\{X\left( i\right)\} = M\left( 1-e^{-W_{2}\Gamma^{\prime}_{th}}\right)\Rightarrow \lim\limits_{N \rightarrow \infty}\frac{1}{N}\sum_{i=1}^{N}X\left( i\right)  = \lim\limits_{N \rightarrow \infty}\frac{1}{N}\sum_{i=1}^{N}M \mathbb{O}_{_{SD}}(i)
\label{inflow} 
\end{align} 
From \eqref{upperbou} and \eqref{inflow} we get :
\begin{align}
\lim\limits_{N \rightarrow \infty}\frac{1}{N}\sum_{i=1}^{N}X\left( i\right)  \geq \lim\limits_{N \rightarrow \infty}\frac{1}{N}\sum_{i=1}^{N}E_{_{R}}\left( i\right)\mathbb{O}_{_{SD}}(i)
\label{balcritical}
\end{align}
In  \eqref{balcritical} when inequality condition holds, the buffer is unstable i.e. $\mathcal{A} > \mathcal{D}$. Then from \eqref{expe} the buffer is in absorbing state,  which is unstable. We can get almost surely $M$ amount of energy from the buffer i.e. $\Pr \left( B\left( i\right)  \geq M \right) \rightarrow 1.$
 In \eqref{balcritical} if  equality holds, then from \eqref{inflow}, \eqref{buffequal} and \eqref{balcritical} we get:
\begin{align}
\lim\limits_{N \rightarrow \infty}\frac{1}{N}\sum_{i=1}^{N}M\mathbb{O}_{_{SD}}(i)=\lim\limits_{N \rightarrow \infty}\frac{1}{N}\sum_{i=1}^{N}X\left( i\right)=\lim\limits_{N \rightarrow \infty}\frac{1}{N}\sum_{i=1}^{N}E_{_{R}}\left( i\right)\mathbb{O}_{_{SD}}(i)
\label{buffbala}
\end{align}
Using $E_R(i)=\min(B(i),M)$,  \eqref{buffbala} gives:
\begin{align}
\lim\limits_{N \rightarrow \infty}\frac{1}{N}\sum_{i=1}^{N}M\mathbb{O}_{_{SD}}(i)=\lim\limits_{N \rightarrow \infty}\frac{1}{N}\sum_{B(i) \geq M} \hspace{-0.3 cm} M\mathbb{O}_{_{SD}}(i)+\lim\limits_{N \rightarrow \infty}\frac{1}{N}\sum_{B(i) < M }\hspace{-0.3 cm}B(i)\mathbb{O}_{_{SD}}(i)
\label{critically}
\end{align}
Clearly, \eqref{critically} holds only when the buffer is at the edge of non-absorbing (but it is at the boundary of absorbing and non-absorbing state). This kind of buffers are termed critically balanced buffers. These buffers may be stable, sub-stable or unstable \cite{N.Zlatanov2013},\cite{R.Loynes}. Usually the Markov chain sequences which are critically balanced are unstable \cite{R.Loynes}. From \eqref{critically} it can be observed that the second term is zero, and that $B(i)>M$ almost surely ($E_{_{R}}\left( i\right)=M$  and  $\Pr \left( B\left( i\right)  \geq M \right) \rightarrow 1. $ Thus $M$ amount of energy is almost always present in the buffer for $\phi_{inc}= 1$. 
 \section{Proof of Theorem \ref{Theorem_Limiting_PDF On-Off}} \label{Appendix_Limiting_PDF On-Off}
  Using an approach similar to that used in Appendix \ref{Appendix_Limiting_PDF} for IBEP,   the CDF of $B\left( i+1\right) $ with IOFP can be obtained from energy storage model  in \eqref{BIonoff} as follows:
 \begin{align}
  \Pr \{B\left( i+1\right) \leq x \}= & \Pr \left\lbrace  \left( B\left( i\right)+X\left( i\right) \right)  \leq x , (\gamma_{_{SD}}(i)\geq \Gamma^{\prime}_{th}) \right\rbrace  \nn \\
  &+\Pr \left\lbrace \left(  B\left( i\right)-M+X\left( i\right) \right)  \leq x , (\gamma_{_{SD}}(i)<\Gamma^{\prime}_{th}), (B(i)\geq M) \right\rbrace  \nn \\ 
  & +\Pr \left\lbrace  \left( B\left( i\right)+X\left( i\right)\right)  \leq x , (\gamma_{_{SD}}(i)<\Gamma^{\prime}_{th}), \left( B(i)<M\right)  \right\rbrace .
  \label{onoffsimp} 
  \end{align}
  With $B(i)=U, X(i)=X $ and $\gamma_{_{SD}}(i)=\gamma_{_{SD}}$,  \eqref{onoffsimp} can be written as:
  \begin{align}
  \Pr \{B\left( i+1\right) \leq x \}= & \Pr \{\left(  U+X\right)  \leq x , (\gamma_{_{SD}} \geq \Gamma^{\prime}_{th}) \}\nn \\
  &+\Pr \{\left(  U-M+X\right)  \leq x , (\gamma_{_{SD}}<\Gamma^{\prime}_{th}), (U \geq M) \} \nn \\ 
  & +\Pr \{\left(  U+X \right) \leq x , (\gamma_{_{SD}}<\Gamma^{\prime}_{th}), (U<M) \}.
  \label{onoffsolv}
  \end{align}
  Let $G_{i}(x)$ be the CDF of $ B(i)$, then $G_{i+1}(x)$ is the CDF of $B(i+1)$ i.e. $\Pr \{B\left( i+1\right) \leq x \}=G_{i+1}(x)$.
  After a long time ($i\rightarrow\infty$), the buffer reaches its steady-state so that $G_{i}(x)=G_{i+1}(x)=G(x).$ 
  In this state, \eqref{onoffsolv} can be simplified as follows:
  \begin{align}
  G(x)= &\displaystyle\int_{y=\Gamma^{\prime}_{th}}^{\infty}\displaystyle\int_{u=0}^{x}F_{X}\left( x-u\right)g\left( u\right)f_{\gamma_{_{SD}}}\left( y\right)  du dy \nn \\
  &+ \int_{y=0}^{\Gamma^{\prime}_{th}}\int_{u=M}^{M+x}F_{X}\left( x-u+M\right)g\left( u\right)f_{\gamma_{_{SD}}}\left( y\right)  du dy \nn \\
  &+\int_{y=0}^{\Gamma^{\prime}_{th}}\int_{u=0}^{\min
  	(x,M)}F_{X}\left( x-u\right)g\left( u\right)f_{\gamma_{_{SD}}}\left( y\right)  du dy.
  \label{onoff11}
  \end{align}
  In \eqref{onoff11} the upper limit of third term integral is $\min(x,M)$ since $U<M$ and $F_{X}(0)=0$.  
  Substituting $ f_{\gamma_{_{SD}}}\left( y\right)=W_{2}\exp\left( -W_{2}y\right) $ in \eqref{onoff11} it can be simplified as in \eqref{limcdfintsonoff}.  Using $f_{\gamma_{_{SD}}}\left( y\right)$, when $x \geq M$, the function $G(x)$ is:
  \begin{align}
  G(x)= & e^{-W_{2}\Gamma^{\prime}_{th}}\int_{u=0}^{x}F_{X}\left( x-u\right)g\left( u\right)du 
  +\left( 1-e^{-W_{2}\Gamma^{\prime}_{th}}\right)\left[\int_{u=M}^{M+x}F_{X}\left( x-u+M\right)g\left( u\right)du \right . \nn \\
  &\left.+ \int_{u=0}^{M}F_{X}\left( x-u\right)g\left(u\right)du \right] \hspace{2.9 cm}     x \geq M
  \label{cdfonoff2}
  \end{align}
  Let $g\left( x\right) $ can be defined as:
  \begin{align}
  g\left( x\right) =
  \begin{cases}
  g_{1}\left( x\right)  \hspace{2.9 cm} 0 \leq x < M \\
  g_{2}\left( x\right)  \hspace{2.9 cm}  x \geq M
  \end{cases}
  \label{limpdfsplit}
  \end{align}
  From \eqref{cdfonoff2} and \eqref{limpdfsplit} the function $G\left( x\right) $ can be written as:
 \begin{align}
 G(x)= & \hspace{0.1 cm} e^{-W_{2}\Gamma^{\prime}_{th}}\int_{u=0}^{M}F_{X}\left( x-u\right)g_{1}\left( u\right)du +e^{-W_{2}\Gamma^{\prime}_{th}}\int_{u=M}^{x}F_{X}\left( x-u\right)g_{2}\left( u\right)du +\nn \\
 &\left( 1-e^{-W_{2}\Gamma^{\prime}_{th}}\right)\left[\int_{u=M}^{M+x}F_{X}\left( x-u+M\right)g_{2}\left( u\right)du + \int_{u=0}^{M}F_{X}\left( x-u\right)g_{1}\left(u\right)du \right] \hspace{0.2 cm}     x \geq M
 \label{cdfonoffpart2}
 \end{align}
 After cancellation of some terms in \eqref{cdfonoffpart2} and differentiating 
  \eqref{cdfonoffpart2} w.r.t. `$x$' on both sides, we get:
 \begin{align}
 g_{2}\left( x\right) = & \int_{u=0}^{M}f_{X}\left( x-u\right)g_{1}\left( u\right)du +e^{-W_{2}\Gamma^{\prime}_{th}}\int_{u=M}^{x}f_{X}\left( x-u\right)g_{2}\left( u\right)du +\nn \\
 &\left( 1-e^{-W_{2}\Gamma^{\prime}_{th}}\right)\int_{u=M}^{M+x}f_{X}\left( x-u+M\right)g_{2}\left( u\right)du  \hspace{1.6 cm}     x \geq M.
 \label{onoffintegr}
 \end{align}
  We  postulate that \eqref{onoffintegr}  has a solution of the form $g_{2}\left( x\right) =ke^{Qx}$ where $k$ is a constant. Using $g_{2}\left( x\right)=ke^{Qx}$ and $f_{X}\left( x\right)=\lambda_{1}e^{-\lambda_{1}x} $, \eqref{onoffintegr} can be simplified as follows:
 \begin{align}
ke^{Qx} = & \hspace{0.1 cm} \lambda_{1}e^{-\lambda_{1}x} \int_{u=0}^{M}e^{\lambda_{1}u}g_{1}\left( u\right)du +e^{-W_{2}\Gamma^{\prime}_{th}}\lambda_{1}e^{-\lambda_{1}x}\int_{u=M}^{x}e^{\lambda_{1}u}ke^{Qu}du +\nn \\
 &\left( 1-e^{-W_{2}\Gamma^{\prime}_{th}}\right)\lambda_{1}e^{-\lambda_{1}x}\int_{u=M}^{M+x}e^{\lambda_{1}\left( u-M\right) }ke^{Qu}du  \hspace{1.6 cm}     x \geq M
 \label{onoffintsol}
 \end{align}
 Simplifying \eqref{onoffintsol} we get:
 \begin{align}
 ke^{Qx} = & \hspace{0.1 cm} \lambda_{1}e^{-\lambda_{1}x}\left[  \int_{u=0}^{M}e^{\lambda_{1}u}g_{1}\left( u\right)du -\frac{ke^{MQ}}{\left( Q+\lambda_{1}\right) }\left[\left( 1-e^{-W_{2}\Gamma^{\prime}_{th}}\right)+e^{M\lambda_{1}}e^{-W_{2}\Gamma^{\prime}_{th}}\right]  \right] \nn \\
 &+\frac{k\lambda_{1}e^{Qx}}{\left(Q+\lambda_{1}\right)}\left[  e^{-W_{2}\Gamma^{\prime}_{th}}+e^{MQ}\left( 1-e^{-W_{2}\Gamma^{\prime}_{th}}\right) \right]
 \hspace{1.6 cm}     x \geq M.
 \label{onoffequat}
 \end{align}
 Both sides of \eqref{onoffequat} are equal when:
 \begin{align}
 \int_{u=0}^{M}e^{\lambda_{1}u}g_{1}\left( u\right)du=\frac{ke^{MQ}}{\left( Q+\lambda_{1}\right) }\left[\left( 1-e^{-W_{2}\Gamma^{\prime}_{th}}\right)+e^{M\lambda_{1}}e^{-W_{2}\Gamma^{\prime}_{th}}\right], \end{align}
 and
 \begin{align}
\frac{Q+\lambda_{1}}{\lambda_{1}}= e^{-W_{2}\Gamma^{\prime}_{th}}+e^{MQ}\left( 1-e^{-W_{2}\Gamma^{\prime}_{th}}\right).
\label{onoffnonequ}
 \end{align}
 Simplifying \eqref{onoffnonequ} we get:
 \begin{align}
 \lambda_{1}\left( 1-e^{-W_{2}\Gamma^{\prime}_{th}}\right)e^{MQ}=Q+\lambda_{1}\left( 1-e^{-W_{2}\Gamma^{\prime}_{th}}\right).
 \label{onoffeqcond}
  \end{align}
  One default solution of \eqref{onoffeqcond} is $Q=0$. However, this is not feasible solution for the limiting PDF  (area under the PDF needs to be unity).
   The other solution of \eqref{onoffeqcond} can be obtained as \eqref{limpdfsplifin} when $\phi_{inc} > 1$. \\
Similarly from \eqref{limcdfintsonoff} and \eqref{limpdfsplit}, when  $0 \leq x<M$, the function $G(x)$ is given by:
  \begin{align}
 G(x)= \hspace{0.1 cm}& e^{-W_{2}\Gamma^{\prime}_{th}}\int_{u=0}^{x}F_{X}\left( x-u\right)g_{1}\left( u\right)du +\left( 1-e^{-W_{2}\Gamma^{\prime}_{th}}\right)\nn \\
 &\left[\int_{u=M}^{M+x}F_{X}\left( x-u+M\right)g_{2}\left( u\right)du + \int_{u=0}^{x}F_{X}\left( x-u\right)g_{1}\left(u\right)du \right] \hspace{0.1 cm}     0 \leq x < M
 \label{cdfonoffpar}
 \end{align}
 After cancellation of some terms in \eqref{cdfonoffpar} we get:
 \begin{align}
  G(x)= & \int_{u=0}^{x}F_{X}\left( x-u\right)g_{1}\left( u\right)du +\nn \\
 &\left( 1-e^{-W_{2}\Gamma^{\prime}_{th}}\right)\left[\int_{u=M}^{M+x}F_{X}\left( x-u+M\right)g_{2}\left( u\right)du  \right] \hspace{0.6 cm}     0 \leq x < M.
 \label{onoffpdfsolv}
 \end{align}
 Differentiating both sides of \eqref{onoffpdfsolv} w.r.t. `$x$' and using $g_{2}\left( x\right) =ke^{Qx}$ and $f_{X}(x)=\lambda_{1}e^{-\lambda_{1}x},$  we get:
  \begin{align}
 g_{1}(x)= & \hspace{0.1 cm} \lambda_{1} \int_{u=0}^{x}e^{-\lambda_{1}\left( x-u\right) }g_{1}\left( u\right)du +\nn \\
 &\left( 1-e^{-W_{2}\Gamma^{\prime}_{th}}\right)\left[\int_{u=M}^{M+x}\lambda_{1}e^{-\lambda_{1}\left( x-u+M\right) }ke^{Qu}du \right] \hspace{0.6 cm}     0<x \leq M
 \label{cdfonoffpart}
 \end{align}
 Simplifying the second integral in \eqref{cdfonoffpart}, we obtain:
 \begin{align}
 g_{1}(x)=\lambda_{1} \int_{u=0}^{x}e^{-\lambda_{1}\left( x-u\right) }g_{1}\left( u\right)du+\frac{k\lambda_{1}e^{QM}\left( 1-e^{-W_{2}\Gamma^{\prime}_{th}}\right)}{\left( Q+\lambda_{1}\right) }\left[ e^{Qx}-e^{-\lambda_{1}x}\right]. 
 \label{onoffvolt}
 \end{align}
 The above  integral equation can be rewritten as follows:
  \begin{align}
 g_{1}(x)=\lambda_{1} \int_{u=0}^{x}e^{-\lambda_{1}\left( x-u\right) }g_{1}\left( u\right)du+r\left( x\right), 
 \label{onoffvolter}
 \end{align}
 where $r(x)=\dfrac{k\lambda_{1}e^{QM}\left( 1-e^{-W_{2}\Gamma^{\prime}_{th}}\right)}{\left(Q+\lambda_{1}\right) }\left[ e^{Qx}-e^{-\lambda_{1}x}\right]$.
\eqref{onoffvolter} is a Volterra type integral equation of the second kind whose solution is given by \cite{R.Morsi2014}:
  \begin{align}
 g_{1}(x)=r\left( x\right)+\lambda_{1}\int_{t=0}^{x}r(t) dt. 
 \label{valofrx}
 \end{align}
 Substituting $r(x)$ in  \eqref{valofrx} we get:
 \begin{align}
 g_{1}(x)=\dfrac{k\lambda_{1}e^{QM}\left( 1-e^{-W_{2}\Gamma^{\prime}_{th}}\right)}{\left( Q+\lambda_{1}\right) }\left[ e^{Qx}-e^{-\lambda_{1}x}\right]+ \lambda_{1}\int_{t=0}^{x}\dfrac{k\lambda_{1}e^{QM}\left( 1-e^{-W_{2}\Gamma^{\prime}_{th}}\right)}{\left(Q+\lambda_{1}\right) }\left[ e^{Qt}-e^{-\lambda_{1}t}\right]dt.
 \label{onoffint}
 \end{align}
  Simplifying \eqref{onoffint} we get:
  \begin{align}
   g_{1}(x)=\frac{k \lambda_{1}e^{QM}\left( 1-e^{-W_{2}\Gamma^{\prime}_{th}}\right)\left( e^{Qx}-1\right)}{Q}.
   \label{funcg1x} 
  \end{align}
  To find $k$  in $g_{2}\left( x\right) $ and $g_{1}\left( x\right) $ we impose  the unit area property of PDF $g(x)$ i.e. 
  \begin{align}
  \int_{x=0}^{\infty}g\left( x\right) dx=\int_{x=0}^{M}g_{1}\left( x\right)dx +\int_{x=M}^{\infty}g_{2}\left( x\right) dx=1
  \end{align}
  \begin{align}
\frac{k\lambda_{1}e^{QM}\left( 1-e^{-W_{2}\Gamma^{\prime}_{th}}\right) }{Q} \int_{x=0}^{M}\left( e^{Qx}-1\right) dx +\int_{x=M}^{\infty}ke^{Qx} dx=1.
\label{onoffkval}
  \end{align}
 Simplifying \eqref{onoffkval}, the value of $k$ is obtained as:
  \begin{align}
  k=\frac{-Q}{M\left(Q+\lambda_{1}\left( 1-e^{-W_{2}\Gamma^{\prime}_{th}}\right)\right)}.
  \end{align}
  Substituting $k$ value in \eqref{funcg1x} and using \eqref{onoffeqcond}  we obtain:
  \begin{align}
  g_{1}\left( x\right) =\frac{1}{M}\left( 1-e^{Qx}\right). 
  \label{funcg1xfin}
  \end{align}
  Similarly after substituting $k$ value in $g_{2}\left( x\right)=ke^{Qx} $ we get:
  \begin{align}
  g_{2}\left( x\right) =\frac{-Qe^{Qx}}{M\left(Q+\lambda_{1}\left( 1-e^{-W_{2}\Gamma^{\prime}_{th}}\right)\right)}.
  \label{funcg2xfin}
  \end{align}
 Substituting \eqref{funcg1xfin} and \eqref{funcg2xfin} in \eqref{limpdfsplit}, we obtain \eqref{limpdfsplifin}.  
 \section{Outage analysis of HSU with IBEP} \label{outage HSU arch}
Simplifying the outage probability expression  \eqref{outin}, the final outage probability expression for incremental relaying with DF protocol is given by \cite{S.S.Ikki}:
\begin{align}
P_{out}= &\Pr\left\lbrace  \gamma_{_{SR}}(i) < \Gamma_{th}\right\rbrace  \Pr\left\lbrace \gamma_{_{SD}}(i) < \Gamma^{\prime}_{th}\right\rbrace \nn \\
&+\Pr\left\lbrace  \gamma_{_{SR}}(i) \geq \Gamma_{th}\right\rbrace \Pr\left\lbrace \gamma_{_{RD}}(i)+\gamma_{_{SD}}(i) < \Gamma_{th} \right\rbrace.
\label{outgen}
\end{align} 
Clearly, 
$\Pr\left\lbrace  \gamma_{_{SR}}(i) < \Gamma_{th}\right\rbrace=1-\exp\left( - W_{4}\Gamma_{th}\right), \hspace{0.1 cm}
\Pr\left\lbrace  \gamma_{_{SD}}(i) < \Gamma^{\prime}_{th}\right\rbrace=1-\exp\left( - W_{2}\Gamma^{\prime}_{th}\right), 
\text{and} \\ \hspace{0.4 cm}\Pr\left\lbrace  \gamma_{_{SR}}(i) \geq \Gamma_{th}\right\rbrace=\exp\left( - W_{4}\Gamma_{th}\right).$ Further:
\begin{align}
\Pr\left\lbrace \gamma_{_{RD}}(i)+\gamma_{_{SD}}(i) < \Gamma_{th} \right\rbrace=\Pr\left\lbrace \frac {P_{_{R}}(i)|h_{_{RD}}(i)|^{2}}{\sigma^{2}}+\gamma_{_{SD}}(i) < \Gamma_{th}
\right\rbrace. 
\label{rel}
\end{align}
Since $P_{_{R}}(i)$ in \eqref{reltrb} depends on $\phi_{inc}$, it will be convenient to study the cases when $\phi_{inc} > 1$ and $\phi_{inc} \leq 1$
 separately.\\
\underline{{\bf Case 1:}} $ \hspace{1 cm} \phi_{inc} > 1 $\\
Substituting the relay transmit power $ P_{_{R}}(i)$ from \eqref{reltrb} into \eqref{rel} gives:
\begin{align}
\Pr\left\lbrace \frac {P_{_{R}}(i)|h_{_{RD}}(i)|^{2}}{\sigma^{2}}+\gamma_{_{SD}}(i) < \Gamma_{th}
\right\rbrace =&\Pr\left\lbrace \frac {2\min\left( B\left( i\right),M\right)|h_{_{RD}}(i)|^{2}}{\sigma^{2}}+\gamma_{_{SD}}(i) < \Gamma_{th}\right\rbrace  \nn \\
=&\Pr\left\lbrace \frac {2 B\left( i\right)|h_{_{RD}}(i)|^{2}}{\sigma^{2}}+\gamma_{_{SD}}(i) < \Gamma_{th},B\left( i\right)<M\right\rbrace \nn \\
+&\Pr\left\lbrace \frac {2 M|h_{_{RD}}(i)|^{2}}{\sigma^{2}}+\gamma_{_{SD}}(i) < \Gamma_{th},B\left( i\right)\geq M\right\rbrace
\label{outpr}
\end{align}
The first term in \eqref{outpr} can be simplified using
$ \Psi_{_{RD}}\left( i\right) \triangleq \dfrac{2|h_{_{RD}}(i)|^{2}}{\sigma^{2}}.$
 Denote the PDFs of  $\Psi_{_{RD}}\left( i\right)$ and $\gamma_{_{SD}}\left( i\right)$
by  $f_{\Psi_{_{RD}}\left( i\right)}\left( y\right)=W_{3}e^{-W_{3}y} $ and $f_{\gamma_{_{SD}}\left( i\right)}\left( z\right)=W_{2}e^{-W_{2}z}, $ and  the limiting PDF of $B\left( i\right)$ by $g\left( x\right)=Z_{1}e^{-Z_{1}x}. $  Then \eqref{outpr} can be simplified as follows:
\begin{align}
&\Pr\left\lbrace  B\left( i\right)\Psi_{_{RD}}\left( i\right) +\gamma_{_{SD}}(i) < \Gamma_{th}, B\left( i\right)<M\right\rbrace \nn \\
&= \int_{x=0}^{M} \int_{y=0}^{\frac{\Gamma_{th}-z}{x}}\int_{z=0}^{\Gamma_{th}}g\left( x\right) f_{\Psi_{_{RD}}\left( i\right)}\left( y\right) f_{\gamma_{_{SD}}\left( i\right)}\left( z\right) dx dy dz \nn \\
&=\left( 1-e^{-Z_{1}M} \right) \left( 1-e^{-W_{2}\Gamma_{th}}\right) 
+Z_{1}W_{2}\int_{x=0}^{M}\frac{xe^{-Z_{1}x}[e^{-W_{2}\Gamma_{th}}-e^{-\frac  {W_{3}\Gamma_{th}}{x}}]}{(xW_{2}-W_{3})}dx,
\label{term1}
\end{align}
where $Z_{1}$ is given by \eqref{limcdf}. \\
The second term in \eqref{outpr} can be simplified  using
$ V_{_{RD}}\left( i\right) \triangleq \dfrac{2M|h_{_{RD}}(i)|^{2}}{\sigma^{2}}.$\\
Denote the PDF of $V_{_{RD}}\left( i\right)$ by $f_{V_{_{RD}}\left( i\right)}\left( y\right)=W_{1}e^{-W_{1}y}. $  Then the second term in \eqref{outpr} can be simplified as follows:
\begin{align}
&\Pr\left\lbrace V_{_{RD}}\left( i\right) +\gamma_{_{SD}}(i) < \Gamma_{th},B\left( i\right) \geq M\right\rbrace \nn \\
&= \int_{x=M}^{\infty} \int_{y=0}^{\Gamma_{th}-z}\int_{z=0}^{\Gamma_{th}}g\left( x\right) f_{V_{_{RD}}\left( i\right)}\left( y\right) f_{\gamma_{_{SD}}\left( i\right)}\left( z\right) dx dy dz \nn \\
&=e^{-Z_{1}M} \left\lbrace \left( 1-e^{-W_{2}\Gamma_{th}}\right) 
+\frac{W_{2}\left[ e^{-W_{1}\Gamma_{th}}-e^{-W_{2}\Gamma_{th}}\right]}{(W_{1}-W_{2})} \right\rbrace .
\label{term2}
\end{align}
Substituting \eqref{term1}, \eqref{term2} into \eqref{outpr} then \eqref{rel}, we get:
\begin{align}
&\Pr\left\lbrace \gamma_{_{RD}}(i)+\gamma_{_{SD}}(i)\leq \Gamma_{th}, \right\rbrace=e^{-Z_{1}M} \left\lbrace \left( 1-e^{-W_{2}\Gamma_{th}}\right) 
+\frac{W_{2}\left[ e^{-W_{1}\Gamma_{th}}-e^{-W_{2}\Gamma_{th}}\right]}{(W_{1}-W_{2})} \right\rbrace \nn \\
&+\left( 1-e^{-Z_{1}M} \right) \left( 1-e^{-W_{2}\Gamma_{th}}\right) 
+Z_{1}W_{2}\int_{x=0}^{M}\frac{xe^{-Z_{1}x}[e^{-W_{2}\Gamma_{th}}-e^{-\frac  {W_{3}\Gamma_{th}}{x}}]}{(xW_{2}-W_{3})}dx.
\label{ergc}
\end{align}
Substituting \eqref{ergc} in \eqref{outgen} we obtain \eqref{hsuoutag1}.\\
\underline{{\bf Case 2:}} $ \hspace{1 cm} \phi_{inc} \leq 1 $\\
Substituting $ P_{_{R}}(i)$ from \eqref{reltrb}  into \eqref{rel} and then \eqref{rel}, we get:
\begin{align}
&\Pr\left\lbrace \gamma_{_{RD}}(i)+\gamma_{_{SD}}(i) < \Gamma_{th}
\right\rbrace =\Pr\left\lbrace \frac {2 M|h_{_{RD}}(i)|^{2}}{\sigma^{2}}+\gamma_{_{SD}}(i) < \Gamma_{th}\right\rbrace\nn \\
&=\int_{y=0}^{\Gamma_{th}-z}\int_{z=0}^{\Gamma_{th}}f_{V_{_{RD}}\left( i\right)}\left( y\right) f_{\gamma_{_{SD}}\left( i\right)}\left(z\right) dy dz 
= \left( 1-e^{-W_{2}\Gamma_{th}}\right) 
+\frac{W_{2}\left[ e^{-W_{1}\Gamma_{th}}-e^{-W_{2}\Gamma_{th}}\right]}{(W_{1}-W_{2})} 
\label{term3}
\end{align}
Substituting \eqref{term3} into \eqref{outgen} we obtain \eqref{hsuoutage2}. 
\section{Outage analysis of HSU with IOFP}\label{outage on-off}
\label{onoff outage infinite}
When $\phi_{inc} > 1$, substituting  $ P_{_{R}}(i)$ from \eqref{onofftra} into \eqref{rel} we get:
\begin{align}
\Pr\left\lbrace \frac {P_{_{R}}(i)|h_{_{RD}}(i)|^{2}}{\sigma^{2}}+  \gamma_{_{SD}}(i) < \Gamma_{th}
\right\rbrace 
=& \Pr\left\lbrace \frac {2 M|h_{_{RD}}(i)|^{2}}{\sigma^{2}}+  \gamma_{_{SD}}(i) < \Gamma_{th},B\left( i\right)\geq M\right\rbrace \nn \\
& +\Pr\left\lbrace   \gamma_{_{SD}}(i)\leq \Gamma_{th},B\left( i\right)<M\right\rbrace.
\label{outpronoff}
\end{align}
We first note that using $g\left( x\right)$ from \eqref{limpdfsplifin} and \eqref{onoffeqcond}, $ \Pr\left\lbrace B\left( i\right)\geq M\right\rbrace$ can be simplified as follows:
\begin{align}
\Pr\left\lbrace B\left( i\right)\geq M\right\rbrace=&\int_{M}^{\infty}g_{2}(x)dx=\int_{M}^{\infty}\frac{-Qe^{Qx}}{M\left( Q+\lambda_{1}\left( 1-e^{-W_{2}\Gamma^{\prime}_{th}}\right)\right)}dx 
\nn \\
=&\frac{e^{QM}}{M\left( Q+\lambda_{1}\left( 1-e^{-W_{2}\Gamma^{\prime}_{th}}\right)\right)}=\frac{e^{QM}}{M\lambda_{1}\left( 1-e^{-W_{2}\Gamma^{\prime}_{th}}\right)e^{QM}}
=\frac{1}{\phi_{inc}}
\end{align}
Let  $ V_{_{RD}}\left( i\right) \triangleq \dfrac{2M|h_{_{RD}}(i)|^{2}}{\sigma^{2}}.$  The first term in \eqref{outpronoff} can be simplified using independence of $B\left( i\right) $, $\gamma_{_{SD}}(i)$ and $|h_{_{RD}}(i)|^{2}$ as follows:
\begin{align}
&\Pr\left\lbrace \frac {2 M|h_{_{RD}}(i)|^{2}}{\sigma^{2}}+  \gamma_{_{SD}}(i) < \Gamma_{th},B\left( i\right) \geq M\right\rbrace 
= \Pr\left\lbrace V_{_{RD}}\left( i\right) +  \gamma_{_{SD}}(i) < \Gamma_{th}\right\rbrace \Pr\left\lbrace B\left( i\right) \geq M\right\rbrace\nn \\
&=\frac{1}{\phi_{inc}} \left\lbrace \left( 1-e^{-W_{2}\Gamma_{th}}\right) 
+\frac{W_{2}\left[ e^{-W_{1}\Gamma_{th}}-e^{-W_{2}\Gamma_{th}}\right]}{(W_{1}-W_{2})} \right\rbrace. 
\label{onoffinfi1}
\end{align}
The second term in \eqref{outpronoff} can be simplified as follows:
\begin{align}
\Pr\left\lbrace   \gamma_{_{SD}}(i) < \Gamma_{th},B\left( i\right)<M\right\rbrace=&\Pr\left\lbrace   \gamma_{_{SD}}(i) < \Gamma_{th}\right\rbrace \Pr\left\lbrace B\left( i\right)<M\right\rbrace \nn \\
=& \left( 1-e^{-W_{2}\Gamma_{th}}\right) \left( 1-\frac{1}{\phi_{inc}}\right) 
\label{onoffinf2}
\end{align}
Substituting \eqref{onoffinfi1} and \eqref{onoffinf2} into  \eqref{outpronoff}, we get: 
\begin{align}
\Pr\left\lbrace   \gamma_{_{RD}}(i)+  \gamma_{_{SD}}(i)< \Gamma_{th}, \right\rbrace= &\hspace{0.1 cm}\frac{1}{\phi_{inc}} \left\lbrace \left( 1-e^{-W_{2}\Gamma_{th}}\right) 
+\frac{W_{2}\left[ e^{-W_{1}\Gamma_{th}}-e^{-W_{2}\Gamma_{th}}\right]}{(W_{1}-W_{2})} \right\rbrace \nn \\
&+\left( 1-\frac{1}{\phi_{inc}} \right) \left( 1-e^{-W_{2}\Gamma_{th}}\right). 
\label{onoffinfout}
\end{align}
Substituting  \eqref{onoffinfout} into \eqref{outgen}, we obtain \eqref{hsuonoffoutag1}.\\ 
\bibliographystyle{IEEEtran}
\bibliography{Journalbbl}

\end{document}